\documentclass[USenglish]{lipics-v2019}

\nolinenumbers

\usepackage{setspace} 
\usepackage{cleveref}
\usepackage{url}
\usepackage{xcolor}
\usepackage{float}
\usepackage{amssymb}
\usepackage{amsmath}
\usepackage[inline]{enumitem}
\usepackage[ruled,vlined,linesnumbered]{algorithm2e}
\usepackage{graphicx,color}
\usepackage{verbatim}
\usepackage{mathtools}
\usepackage{bbding}
\usepackage{pifont}
\usepackage[labelformat=simple]{subcaption} 
\usepackage{array}
\usepackage{makecell}
\usepackage{wrapfig}

\newcommand{\RNum}[1]{\uppercase\expandafter{\romannumeral #1\relax}}

\newtheorem{conjecture}{Conjecture}

\newcommand {\tablecell}[1]
{
  \begin {minipage} {3.5cm}
  \centering
  \vspace {.3em}
  #1
  \vspace {.3em}
  \end {minipage}
}
\newcommand {\tablerow}[4]
{
  #1
  & \tablecell{#2}
  & \tablecell{#3}
  & \tablecell{#4}
  \\
}

\newcommand{\Nats}{{\mathbb{N}}}             
\newcommand{\Reals}{{\mathbb{R}}}            
\newcommand{\eps}{\varepsilon}               

\newcommand{\polylog}{\mathrm{polylog}}
\newcommand{\poly}{\mathrm{poly}}

\newcommand{\FSD}{\mathsf{FSD}_{\delta}}
\newcommand{\strip}{\mathsf{ST}_{\delta}}
\newcommand{\Strip}{\mathsf{ST}}
\newcommand{\spine}{\mathsf{SP}_{\delta}}
\newcommand{\Spine}{\mathsf{SP}}
\newcommand{\slist}{\mathsf{L}_\delta}

\newcommand{\partition}{\mathsf{Prt}}

\newcommand{\Frechet}{\mathsf{F}}
\newcommand{\dFrechet}{\mathsf{dF}}
\newcommand{\wFrechet}{\mathsf{wF}}
\newcommand{\Hausdorff}{\mathsf{H}}
\newcommand{\diHausdorff}{\overrightarrow{\mathsf{H}}}
\newcommand{\diHausdorffP}{\overleftarrow{\mathsf{H}}}
\usepackage{mathtools}
\DeclarePairedDelimiter\segment {\langle}{\rangle}
\newcommand{\Frd}{Fr\'echet distance}
\newcommand{\etal}{{\it et al.\xspace}}

\newcommand{\highlight}[1]{\textcolor{darkblue}{#1}}
\newcommand{\tinyspace}{\hspace{1px}}

\def\C{{\cal C}}

\def\G{{\cal G}}

\def\L{{\cal L}}

\def\N{{\cal N}}

\def\P{{\cal P}}

\def\R{{\cal R}}

\def\V{{\cal V}}

\newcommand{\figref}[1]{Figure~\ref{#1}}

\newcommand{\lemref}[1]{Lemma~\ref{#1}}

\newcommand{\corref}[1]{Corollary~\ref{#1}}
\newcommand{\thmref}[1]{Theorem~\ref{#1}}
\newcommand{\secref}[1]{Section~\ref{#1}}

\newcommand{\algref}[1]{Algorithm~\ref{#1}}

\definecolor{lgray}{gray}{0.5}
\definecolor{darkblue}{rgb}{0, 0, 0.7}

\newcommand{\myemph}[1]{\emph{#1}}
\newcommand{\mathemph}[1]{#1}

\title{Global Curve Simplification}

\author{Mees van de Kerkhof}{Utrecht University, the Netherlands}{m.a.vandekerkhof@uu.nl}{}{Supported by the Netherlands Organisation for Scientific Research (NWO) under project number 628.011.005.}
\author{Irina Kostitsyna}{TU Eindhoven, the Netherlands}{i.kostitsyna@tue.nl}{}{}
\author{Maarten L\"offler}{Utrecht University, the Netherlands}{m.loffler@uu.nl}{}{Partially supported by the Netherlands Organisation for Scientific Research (NWO) under project numbers 614.001.504 and 628.011.005.}
\author{Majid Mirzanezhad}{Tulane University, USA}{mmirzane@tulane.edu}{}{Supported by the National Science Foundation grant CCF-1637576.}
\author{Carola Wenk}{Tulane University, USA}{cwenk@tulane.edu}{}{Supported by the National Science Foundation grant CCF-1637576.}

\authorrunning{M. van de Kerkhof, I. Kostitsyna, M. L\"offler, M.~Mirzanezhad and C. Wenk} 

\Copyright{Mees van de Kerkhof, Irina Kostitsyna, Maarten L\"offler, Majid Mirzanezhad and Carola Wenk}

\hideLIPIcs

\begin{document}
\maketitle

\begin{abstract}
Due to its many applications, {\em curve simplification} is a long-studied problem in computational geometry and adjacent disciplines, such as graphics, geographical information science, etc.
Given a polygonal curve $P$ with $n$ vertices, the goal is to find another polygonal curve $P'$ with a smaller number of vertices such that $P'$ is sufficiently similar to $P$.
Quality guarantees of a simplification are usually given in a {\em local} sense, bounding the distance between a shortcut and its corresponding section of the curve. 
In this work we aim to provide a systematic overview of curve simplification problems under {\em global} distance measures that bound the distance between $P$ and $P'$. 
We consider six different curve distance measures: three variants of the {\em Hausdorff} distance and three variants of the \emph{Fr\'echet} distance. 
And we study different restrictions on the choice of vertices for $P'$. 
We provide polynomial-time algorithms for some variants of the global curve simplification problem, and show NP-hardness for other variants. Through this systematic study we observe, for the first time, some surprising patterns, and suggest directions for future research in this important area.

\end{abstract}


\section{Introduction}\label{sec:intro}
Due to its many applications, {\em curve simplification} (also known as {\em line simplification}) is a long-studied problem in computational geometry and adjacent disciplines, such as graphics, geographical information science, etc.
Given a polygonal curve $P$ with $n$ vertices, the goal is to find another polygonal curve $P'$ with a smaller number of vertices such that $P'$ is sufficiently similar to $P$.
Classical algorithms for this problem famously include a simple recursive scheme by Douglas and Peucker~\cite{douglas73algorithms}, and a more involved dynamic programming approach by Imai and Iri~\cite{imai88algorithms}; both are frequently implemented and cited.
Since then, numerous further results on curve simplification, often in specific settings or under additional constraints, have been obtained~\cite {abam10streaming,ahmw-nltaa-05,barequet2002approx3,berg98correct, buzer07optimal,cc-apcwmnls-96,chen05angle,g-nmfc-96,ii-aoaaplf-86}.

Despite its popularity, the Douglas-Peucker algorithm comes with no provable quality guarantees. The method by Imai and Iri, though slower, was introduced as an alternative which does supply guarantees: it finds an optimal shortest path in a graph in which potential shortcuts are marked as either {\em valid} or {\em invalid}, based on their distance to the corresponding sections of the input curve.
However, Agarwal \etal~\cite {ahmw-nltaa-05} note that the Imai-Iri algorithm does not actually globally optimize any distance measure between the original curve $P$ and the simplification $P'$. This work initiated a more formal study of curve simplification; van Kreveld \etal~\cite {klw-oopsihfd-18} systematically show that both Douglas-Peucker and Imai-Iri may indeed produce far-from-optimal results.

This raises a question of what it means for a simplification to be optimal.
We may view it as a dual-optimization problem: we wish to minimize the number of vertices of $P'$ given a constraint on its similarity to $P$. This depends on the distance measure used; popular curve distance measures include the {\em Hausdorff} and {\em Fr\'echet} distances (variants and formal definitions are discussed in \secref{sec:distance_measures}).



However, the difference in interpretation between Agarwal \etal~and Imai and Iri lies not so much in the choice of distance measure, but rather what exactly the measure is applied to.
In fact, the Imai-Iri algorithm is optimal in a {\em local} sense: it outputs a subsequence of the vertices of $P$ such that the Hausdorff distance between each shortcut and {\em its corresponding section of the input} is bounded: each shortcut approximates the section of $P$ between the vertices of the shortcut.

In this work, we underline this difference by using the term {\em global} simplification when a bound on a distance measure must be satisfied between $P$ and $P'$ (formal definition in \secref{sec:GCSOverview}), and {\em local} simplification when a bound on a distance measure must be satisfied between each edge of $P'$ and its corresponding section of $P$.
Clearly, a local simplification is also a global simplification, but the reverse is not necessarily true, see Figure~\ref {fig:globallocal}.
\begin{figure} [b]
	\centering \includegraphics {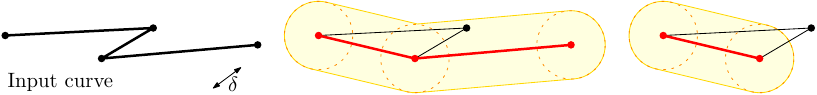} 
	\caption{For a target distance $\delta$, the red curve (middle) is a global simplification of the input curve (left), but it is not a local simplification, since the first shortcut does not closely represent its corresponding curve section (right). The example works for both Hausdorff and Fr\'echet distance.}
	\label{fig:globallocal}
\end{figure}

Both local and global simplifications have their merits: one can imagine situations where it is important that each segment of a simplified curve is a good representation of the curve section it replaces, but in other applications (e.g., visualization) it is really the similarity of the overall result to the original that matters.
Most existing work on curve simplification falls in the {\em local} category.
In this we focus on the curve simplification when $D(.,.)$ is considered a global distance measure. For more brevity we call the problem {\em GCS}.

\subsection {Existing Work on Global Curve Simplification} \label {sec:related}


Surprisingly, only a few results on simplification under global distance measures are known~\cite {ahmw-nltaa-05, bjwyz-s3pcudfd-08, bc-pscc-18, klw-oopsihfd-18}; consequently, what makes the problem difficult is not well understood.

Agarwal~\etal~\cite {ahmw-nltaa-05} first consider the idea of global simplification. They introduce what they call a {\em weak simplification}: a model in which the vertices of the simplification are not restricted to be a subset of the input vertices, but can lie anywhere in the ambient space.\footnote {We choose not to adopt the terms {\em weak} and {\em strong} in this context because we will also distinguish an intermediate model, and to avoid confusion with the {\em weak Fr\'echet} distance; refer to Section~\ref {sec:vertex_restrictions}.}
Interestingly, they compare this to a {\em local} simplification where vertices are restricted to be a subset of the input.
We may interpret a combination of two of their results (Theorem 1.2 and Theorem 4.1) as an approximation-algorithm for global curve simplification with unrestricted vertices under the Fr\'echet distance: for a given curve $P$ and threshold $\delta$ one can compute, in $O(n \log n)$ time, a simplification $P'$ which has at most the number of vertices of an optimal simplification with threshold $\delta/8$.

Bereg~\etal~\cite{bjwyz-s3pcudfd-08} first explicitly consider global simplification in the setting where vertices are restricted to be a subsequence of input vertices, but using the {\em discrete Fr\'echet distance}: a variant of the Fr\'echet distance which only measures distances between vertices (refer to Section~\ref {sec:distance_measures}).
They show how to compute an optimal simplification where vertices are restricted to be a subsequence of the vertices in $P$ in $O(n^2)$ time, and they give an $O(n\log n)$ time algorithm for the setting where vertices may be placed freely.


Van Kreveld~\etal~\cite{klw-oopsihfd-18} consider the same (global distance, but vertices should be a subsequence) setting, but for the continuous Fr\'echet and Hausdorff distances. They give polynomial-time algorithms for the Fr\'echet distance and directed  Hausdorff distance (from simplification curve to input curve), but they show the problem is NP-hard for the directed Hausdorff distance in the opposite direction and for the undirected Hausdorff distance.
Recently, Bringmann and Chaudhury \cite{bc-pscc-18} improved their result for the Fr\'echet distance when the vertices in $P'$ are a subsequence of  to $O(n^3)$, and also give a conditional cubic lower bound.

Finally, we mention there is earlier work which does not explicitly study simplification under global distance measures, but contains results that may be reinterpreted as such.
Guibas \etal~\cite{ghms-apswmlp-93} provide algorithms for computing minimum-link paths that stab a sequence of regions in order. One of the variants, presented in Theorems 10 and 14 of ~\cite{ghms-apswmlp-93}, computes what may be seen as an optimal simplification under the Fr\'echet distance with no vertex restrictions, i.e., the same setting that was studied by Agarwal~\etal, in $O(n^2\log^2 n)$ time in the plane.
%




In Section~\ref {sec:class}, we present a formal classification of global curve simplification problems.
Table~\ref {tab:results} gives an overview of known results, as well as several new results to complement these (in some cases straightforward adaptations of known results).
In the remainder of the paper we sketch the main ideas behind our new results.

\section{Classification}\label{sec:class}

We aim to provide a systematic overview of curve simplification problems under global distance measures. To this end, we have collected known results and arranged them in a table (Table~\ref {tab:results}), and provided several new results to complement these (refer to Section~\ref {sec:results}). This allows us for the first time to observe some surprising patterns, and it suggests directions for future research in this important area.
We first discuss the dimensions of the table.

\subsection {Distance Measures} \label {sec:distance_measures}

For our study, we consider six different curve distance measures: three variants of the {\em Hausdorff} distance and three variants of the \emph{Fr\'echet} distance.
These are among the most popular curve distance measures in the algorithms literature.
The Hausdorff distance captures the maximum distance from a point on one curve to a point on the other curve.
The variants of the Hausdorff distance we consider are the directed Hausdorff distance from the input to the output, the directed Hausdorff distance from the output to the input, and the undirected (or bidirectional) Hausdorff distance.
The Fr\'echet distance captures the maximum distance between a pair of points traveling along the two curves simultaneously without moving backward.
We now formally define all six distance measures.

Let ${P=\langle p_1,p_2,\cdots, p_n \rangle}$ be the input
polygonal curve.
We treat $P$ as a continuous
map $P:[1,n] \rightarrow \mathbb{R}^d$, where $P(i)=p_i$ for integer $i$,
and the $i$-th edge is linearly parametrized as $P(i + \lambda) =
(1-\lambda) p_i + \lambda p_{i +1}$.
We write $P[s,t]$ for the subcurve between $P(s)$ and $P(t)$ and denote the {\em shortcut}, i.e., the  straight line connecting them, by $\segment{P(s)P(t)}$.

The \emph{\Frd} between two polygonal curves $P$ and $Q$, with $n$ and $m$ vertices, respectively, is 
${\Frechet(P, Q)}= \inf_{(\sigma,\theta)} \max_{t} \| P(\sigma(t))-Q(\theta(t))\|$, where $\sigma$ and $\theta$ are continuous non-decreasing functions from $[0,1]$ to $[1,n]$ and $[1,m]$, respectively.
If $\sigma$ and $\theta$ are continuous but not necessarily monotone, the resulting infimum is called the \emph{weak \Frd}.
Finally, the {\em discrete \Frd} is a variant where $\sigma$ and $\theta$ are discrete functions from $\{1, \ldots, k\}$ to $\{1, \ldots, n\}$ and $\{1, \ldots, m\}$ with the property that $|\sigma(i)-\sigma(i+1)|\le 1$.

The \emph{directed Hausdorff distance} between two polygonal curves (or more generally, compact sets)
$P$ and $Q$ is defined as
${\diHausdorff(P,Q)}=\max\limits_{p \in
	P} \min\limits_{q \in Q}\|p-q\|$.
The \emph{undirected Hausdorff distance} is then simply the maximum over the two directions:
${\Hausdorff(P,Q)}= \max \{\diHausdorff(P,Q),\diHausdorff(Q,P)\}$.

\subsection {Vertex Restrictions} \label {sec:vertex_restrictions}

\begin{figure}
	\centering
	\includegraphics {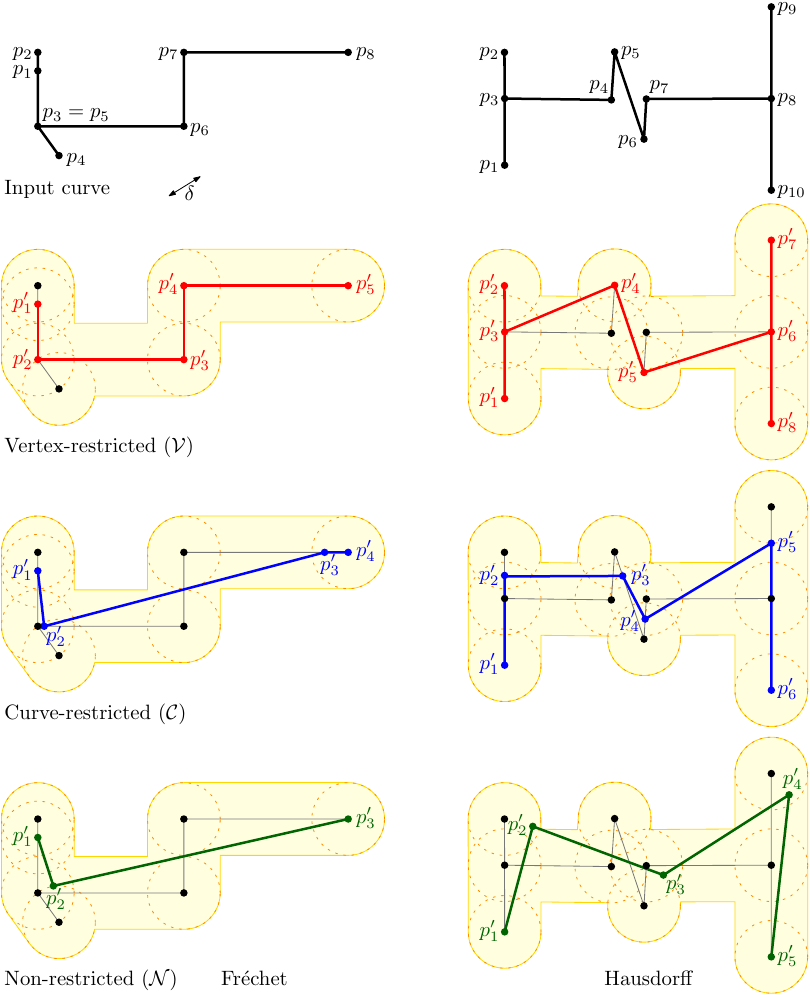}
	\caption{
		The GCS simplified curve in different restrictions under Fr\'echet (left) and Hausdorff (right) distances. 
		The vertex-restricted case (in red) requires 5 vertices for \Frd\ and 8 vertices for the Hausdorff distance. The curve-restricted case (in blue) requires 4 vertices for \Frd\ and 6 vertices for the Hausdorff distance. The non-restricted case (in green) requires only 3 vertices for \Frd\ and only 5 vertices for the Hausdorff distance.
	}
	\label{fig:restrictions}
\end{figure}

Once we have fixed the distance measure and agreed that we wish to apply it globally, one important design decision still remains to be made.  Traditional curve simplification algorithms consider the (polygonal) input curve $P$ to be a sequence of points, and produce as output $P'$  a subsequence of this sequence.  However, if we measure the distance globally, there may be no strong reason to restrict the family of acceptable output curves so much: the distance measure already ensures the similarity between input and output curves, so perhaps we may allow a more free choice of vertex placement.  Indeed, several results under this more relaxed viewpoint exist, as discussed in Section~\ref {sec:related}.
Here, we choose to investigate three increasing levels of freedom:
(1) \emph{vertex-restricted $({\V})$}, where vertices of $P'$ have to be a subsequence of vertices of $P$;
(2) \emph{curve-restricted (${\C}$)}, where vertices of $P'$ can lie anywhere on $P$ but have to respect the order along $P$; and 
(3) \emph{non-restricted (${\N}$)}, where vertices of $P'$ can be anywhere in the ambient space. 
Figure~\ref {fig:restrictions} illustrates the difference between the three models.
The third category does not make sense for local curve simplification, but is very natural for global curve simplification.
Observe that when the vertices of a simplified curve have more freedom, the optimal simplified curve never has more, but may have fewer, vertices.

\subsection {Global Curve Simplification Overview}\label{sec:GCSOverview}

We are now ready to formally define a class of global curve simplification problems.  When $D(\cdot,\cdot)$ denotes a distance measure between curves (e.g., the \emph{Hausdorff} or \emph{Fr\'echet} distance), the \emph{global curve simplification (GCS)} problem asks what is the smallest number $k$ such that there exists a curve $P'$ with at most $k$ vertices, chosen either as a subsequence of the vertices of $P$ (variant ${\V}$), as a sequence of points on the edges of $P$ in the correct order along $P$ (variant ${\C}$), or chosen anywhere in $\Reals^d$ (variant ${\N}$) and such that $D(P,P') \le \delta$, for a given threshold $\delta$.  In all cases, we require that $P$ and $P'$ start at the same point and end at the same point.

Table~\ref {tab:results} summarizes results for the different variants of the GCS problem obtained by instantiating $D$ with the Hausdorff or Fr\'echet distance measures and by applying a vertex restriction $R$.
Here $R \in \{\V, \C, \N\}$, and $D$ is either the undirected Hausdorff distance $\Hausdorff$, the directed Hausdorff distance $\diHausdorffP(P,\delta)$ from $P$ to $P'$, the directed Hausdorff distance $\diHausdorff(P,\delta)$ from $P'$ to $P$, the \Frd\ $\Frechet$, the discrete \Frd\ $\dFrechet$, or the weak \Frd\ $\wFrechet$. 
Throughout the paper we use $D_R(P,\delta)$ to denote a curve $P'$ that is the optimal $R$-restricted simplification of $P$ with $D(P,P')\leq\delta$.

Since GCS is a dual-optimization problem, we call an algorithm an {\em $(\alpha,\beta)$-approximation} if it computes a solution with distance at most $\beta\delta$ and uses at most $\alpha$ times more shortcuts than the optimal solution for distance $\delta$.

\begin{table}[!ht]
	\centering 
	\begin{tabular}[c]{|c|c|c|c|}
		\hline 
		Distance & Vertex-restricted $(\V)$ & Curve-restricted $(\C)$ & Non-restricted ($\N$)\\ 
		\hline \hline
		\tablerow{$\diHausdorffP(P,\delta)$}
		{ 
			strongly NP-hard~\cite{klw-oopsihfd-18}
		}
		{ 
			weakly NP-hard 
			\highlight{(Thm \ref{thm:dihaushard})}
		}
		{ 
			?
		}
		\hline
		\tablerow{$\diHausdorff(P,\delta)$}
		{ 
			$O(n^4)$~\cite{klw-oopsihfd-18} \\
			O($n^2\polylog\ n$)
			\highlight{(Thm \ref{thm:VRHP'toP})}
		}
		{ 
			weakly NP-hard 
			\highlight{(Thm \ref{thm:CR3D})}
		}
		{ 
			$\poly(n)$ {\cite{klps-ocmlpp-17}}
		}
		\hline
		\tablerow{$\Hausdorff(P,\delta)$}
		{ 
			strongly NP-hard~\cite{klw-oopsihfd-18}
		}
		{ 
			strongly NP-hard
			\highlight{(Cor \ref{cor:CRHD})}
		}
		{ 
			strongly NP-hard 
			\highlight{(Thm \ref{thm:NRHD})}
		}
		\hline
		\tablerow{$\Frechet(P,\delta)$}
		{ 
			$O(mn^5)$~\cite{klw-oopsihfd-18} \\
			$O(n^3)$~\highlight{(Thm \ref{thm:VRFD_improved})} \\
			$O(n^3)$~\cite{bc-pscc-18}
		}
		{ 
			$O(n)$ in $\Reals^1$ \highlight{(Thm \ref{thm:CRFD1})} \\
			weakly NP-hard in $\Reals^2$ \highlight{(Thm \ref{thm:unipinhole})}
		}
		{ 
			$O(n^2\log^2n)$ in $\Reals^2$~\cite{ghms-apswmlp-93} \\
			$O(n \log n)$ $(1,8)$-approx~\cite {ahmw-nltaa-05} \\
			$O^*({n^2\log n \log \log n})$
			$(2,1+\eps)$-approx \highlight{(Thm \ref{thm:NRFD})}
		}
		\hline
		\tablerow{$\dFrechet(P,\delta)$}
		{ 
			$O(n^2)$ \cite{bjwyz-s3pcudfd-08}
		}
		{
			$O(n^3)$ \highlight{(Thm \ref{thm:CRDF})} 
		}
		{
			$O(n\log n)$~\cite{bjwyz-s3pcudfd-08}
		}
		\hline
		\tablerow{$\wFrechet(P,\delta)$}
		{ 
			$O(n^3)$ \highlight{(Thm \ref{thm:VRWFD})}
		}
		{ 
			weakly NP-hard \highlight{(Thm \ref{thm:CR3D})} 
		}
		{ 
			
			$(2,1+\eps)$-approx \highlight{(Cor \ref{cor:NRWFD})}
		}
		\hline
	\end{tabular}
	\caption{Known and new results for the GCS problem under global distance measures. }
	\label{tab:results}
\end{table}

\subsection{New Results} \label {sec:results}

In order to provide a thorough understanding of the different variants of the GCS problem we provide several new results.  In some cases these are straightforward adaptations of known results, in other cases they require deeper ideas.
In \secref{subsec:FrechetDP}, we give
polynomial time algorithms for the vertex-restricted GCS problem under
strong Fr\'echet (\thmref{thm:VRFD_improved}) and weak Fr\'echet
(\thmref{thm:VRWFD}) distance%
\footnote {
	An algorithm with a running time of $O(n^4)$ for the vertex-restricted variant under the strong \Frd~is presented in \thmref{thm:VRFD}.
	Bringmann and Chaudhury \cite{bc-pscc-18} independently developed an $O(n^3)$ algorithm
	for the same problem. Our current algorithm in
	\secref{subsec:FrechetDP} 
	uses a vital insight from~\cite{bc-pscc-18} to improve our $O(n^4)$ algorithm to $O(n^3)$ as well. The resulting algorithm now uses less space than the algorithm in~\cite{bc-pscc-18} and can be extended to the weak Fr\'echet distance.
}. In Section~\ref{sec:VRHausP'toP} we show that the vertex-restricted problem under the directed Hausdorff distance from $P'$ to $P$ considered by van Kreveld \etal~\cite{klw-oopsihfd-18} can be improved to $O(n^3\log n)$ (\thmref{thm:VRHP'toP}). 
%
In \secref{sec:CR} we prove that solving the curve-restricted
GCS is NP-hard for almost all measures regarded in
this paper except for the discrete \Frd\
(\thmref{thm:CRDF}) and strong \Frd\ in $\Reals^1$
(\thmref{thm:CRFD1}) for which we present polynomial time
algorithms. To the best of our knowledge, these are the first results
in the curve-restricted setting under global distance
measures. 
Finally, in Section~\ref{sec:NRFD},  we give $(2,1+\eps)$-approximation algorithm for
$\Frechet_\N(P,\delta)$ which runs
in $O^*(n^2\log n\log \log n)$ time where $O^*$ hides
polynomial factors of $1/\eps$ (\thmref{thm:NRFD}). We also argue
that the same result holds for weak \Frd\
(\corref{cor:NRWFD}). In Section~\ref{sec:NRHD} we show that this problem becomes NP-hard when we
consider the Hausdorff distance (\thmref{thm:NRHD}).

\subsection {Discussion} \label {sec:discuss}





With both the existing work and our new results in place, we now have a good overview of the complexity of the different variants of the GCS problem, see Table~\ref {tab:results}.


Observe that the curve-restricted variants seem to generally be harder than both the vertex-restricted and the non-restricted variants.
That means that, on the one hand, broadening the search space from the vertex-restricted to the curve-restricted case makes the problem harder. But on the other hand it does not give unrestricted freedom of choice, which in turn enables the development of efficient algorithms for the unrestricted case.

Another interesting pattern can be observed for the Hausdorff distance measures. The direction of the Hausdorff distance makes a significant difference in whether the corresponding GCS problem is NP-hard or polynomially solvable.  The GCS problem for the undirected Hausdorff distance is at least as hard as for the directed Hausdorff distance from the input curve to the simplification.

Drawing upon the above observations we make the following conjecture:
\begin{conjecture}
	The curve-restricted and non-restricted GCS problems for $\diHausdorffP(P,\delta)$ are strongly NP-hard.
\end{conjecture}

Also, note that we only prove several problem variants are NP-hard. The question of whether these problems are also in NP remains open.
It is interesting to note that for the problem of computing a minimum-link path inside a simple polygon, it has been shown that coordinates with exponential bit complexity are sometimes required~\cite{klps-ocmlpp-17}.
This suggests the problems we discuss may have a similar structure.

%


\section{Freespace-Based Algorithms for Fr\'echet Simplification} \label{sec:VRFD}
We use the free space diagram  between $P$ and its shortcut graph $G$ to solve the vertex-restricted GCS problem under the weak and strong Fr\'echet distances in $O(n^3)$ time and space.
%
This is related to {\em map-matching} \cite{aerw-mpm-03}, however in our case we need to compute {\em shortest} paths in the free space that correspond to {\em simple} paths in $G$. While map-matching for closed simple paths is NP-complete \cite{m-mmsc-13}, we exploit the DAG property of $G$ to develop efficient algorithms. 

\subsection{Shortcut DAG and Free Space Diagram}

For a given polygonal curve $P$, we define its 
\myemph{shortcut DAG}
$G=G(P)=(V,E)$, where $V=\lbrace 1, \ldots,n\rbrace$ and 
$E=\lbrace (u,v)~|~ 1 \leq u < v \leq n\rbrace$.
We consider each $v\in V$ to be embedded at
$p_v$ and each edge $e=(u,v)\in E$ to be embedded as a straight line
shortcut is linearly parameterized as $e(t)=(1-t)p_u+tp_v$ for $t\in[0,1]$.
We consider the parameter space of $G$ to be $E\times [0,1]$.

%
%
%
%
%

Now, let $\delta>0$, and consider the joint parameter space $[1,n]\times E\times [0,1]$ of $P$ and $G$. Any $(s,e,t)\in [1,n]\times E\times [0,1]$ is called \myemph{free}
if $\|P(s)-e(t)\|\leq \delta$, and the union of all free points is
referred to as the \myemph{free space}.
For brevity, we write $(s,e(t))$ instead of $(s,e,t)$, and if $e(t)=v\in V$ we write $(s,v)$.
The \myemph{free space diagram} $\mathemph{\FSD(P,G)}$ consists of all points in 
$[1,n]\times E\times [0,1]$ together with an annotation for each point whether it is free or not.
%
In the special case that the graph is a
polygonal curve $Q$ with $m$ vertices, then $\FSD(P,Q)$ consists of
$(n-1)\times (m-1)$ cells in the domain $[1,n]\times[1,m]$. A monotone
(continuous) path from $(1,1)$ to $(n,m)$ that lies entirely within the free space
corresponds to a pair of monotone re-parameterizations
$(\sigma,\theta)$ that witness $\Frechet(P, Q)\leq\delta$. Alt and
Godau showed that such a \myemph{reachable path} can be computed in
$O(mn)$ time by propagating reachable points across free space cell
boundaries in a dynamic programming
manner \cite{ag-cfdb-95}.
A path from  $(1,1)$ to $(n,m)$ that lies entirely within the free space
but that is not necessarily monotone corresponds to a pair of re-parameterizations
$(\sigma,\theta)$ that witness $\wFrechet(P, Q)\leq\delta$.
See \figref{fig:gcfs} for an example free
space diagram for two curves.

%
%
%
%
\begin{figure}[t]
	\begin{center}
		\includegraphics[width=10cm]{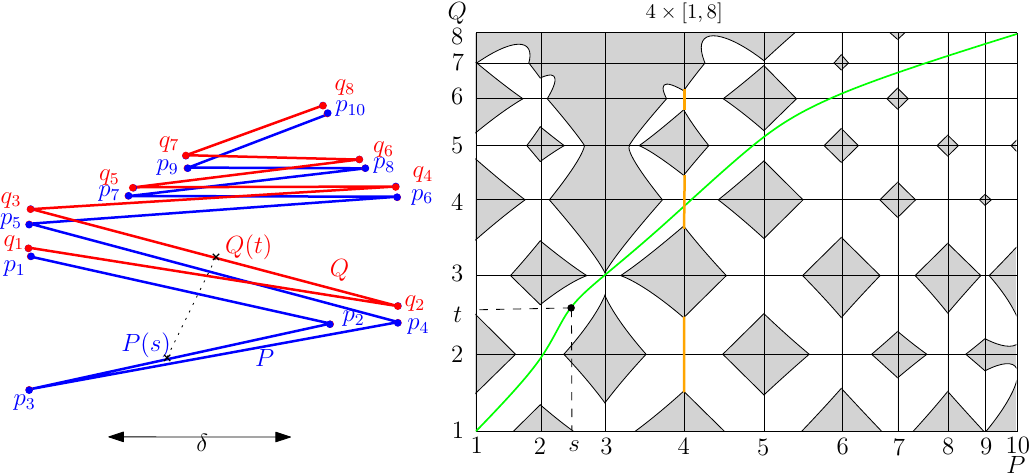}
		\caption{$(s,t)$ is a free point on a reachable path in $\FSD(P,Q)$, free space is shown in white.}
		\label{fig:gcfs}
	\end{center}
\end{figure}

The free space diagram $\FSD(P,G)$ consists of one \myemph{cell} for each edge in $P$ and each edge in $G$. The free space in such a cell is convex. The boundary of a cell consists of four line
\begin{wrapfigure}{r}{.3\textwidth}
	\centering
	\includegraphics[width=.30\textwidth]{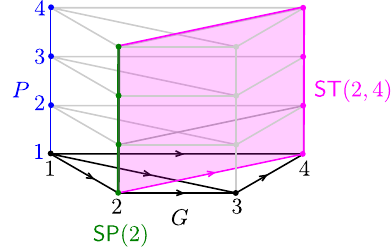}
\end{wrapfigure}
segments which each contain at most a single \myemph{free space interval}. 
We consider the free space diagram to be composed of spines and strips; see the figure on the right.
For any $v\in V$ and $e\in E$ we call $\mathemph{\Spine(v)}=[1,n] \times v$ a \myemph{spine} and $\mathemph{\Strip(e)}=[1,n] \times e\times [0,1]$ a \myemph{strip}.
And we denote the free space within spines and strips as
$\mathemph{\spine(v)}=\lbrace (s,v) ~|~ 1\leq s\leq n,\; ||P(s)-p_v||\leq\delta\rbrace$
and
$\mathemph{\strip(e)}=\lbrace (s,e(t)) ~|~ 1\leq s\leq n,\; 0\leq t\leq 1,\; ||P(s)-e(t)||\leq \delta\rbrace$, respectively.
For any edge $(u,v)\in E$, both spines centered at the vertices of the
edge are subsets of the strip:
$\Spine(u), \Spine(v)\subseteq \Strip(u,v)$, and $\Spine(u)$ is a subset of all strips with respect to edges incident on $u$.
%

\subsection{Weak Fr\'echet Distance $\wFrechet_{\V}(P,\delta)$ in Polynomial Time } \label {subsec:weakvertex}
In this section we provide an algorithm to solve the
vertex-restricted GCS problem under the weak Fr\'echet distance
in  $O(n^3)$ time and space. We will see that the optimal simplification $P'$ is a 
shortest simple path in $G$ that corresponds to a path $\P$ from $(1,1)$ to $(n,n)$ in $\FSD(P,G)$ that is contained in free space.

Let $P'=\wFrechet_V(P,\delta)$ be an optimal vertex-restricted simplification of $P$ and let $n'=\#P'$ be the number of vertices in $P'$.
Then $P'$ is a path in $G$, and $P'$
visits an increasing subsequence of vertices in $P$ (or $V$).
From the fact that $\wFrechet(P,P')\leq\delta$ we know that there is
a path $\P=(\sigma,\theta)$ from $(1,1)$ to $(n,n')$ in $\FSD(P,P')$ that lies entirely within free space.
And since $\FSD(P,P')$ is a subset of $\FSD(P,G)$,  the path $\P=(\sigma,\theta)$ is also a path in $\FSD(P,G)$.
Here, $\sigma$ is a reparameterization of $P$, and $\theta$ is a
reparameterization of a {\em simple} path $P'=\langle p_{i_1},
p_{i_2}, \ldots, p_{i_k}\rangle$ in $G$ with $i_1=1$ and $i_k=n$. We
call $(s,d)$ in $\FSD(P,G)$ \myemph{weakly reachable} if there
exists a path $\P=(\sigma,\theta)$ from $(1,1)$ to $(s,d)$ in
$\FSD(P,G)$ that lies in free space such that $\theta$ is a
reparameterization of a simple path from $p_{1}$ to some point on an
edge in $G$. We denote the number of vertices in this simple path by
$\# \P$, and we call $\P$ \myemph{weakly reachable}.
We define the cost function $\phi: [1,n] \times V\rightarrow \Nats$
as $\mathemph{\phi(z,v)}=\min_{\P} \# \P$, where the minimum ranges over all weakly reachable paths to $(z,v)$ in the free space diagram.  If no such path exists then $\phi(z,v)=\infty$.
Note that all points in a free space interval (on the boundary of a free space cell) have the same $\phi$-value.

\begin{observation}
	There is a weakly reachable path $\P$ in $\FSD(P,G)$ from $(1,1)$ to $(n,n)$
	with $\#\P=\#\wFrechet_\V(P,\delta)$
	if and only if
	$\phi(n,n)=\#\wFrechet_\V(P,\delta)$.
	
\end{observation}

Since $\phi(n,n)=\#\wFrechet_\V(P,\delta)$ is the length of a shortest
simple path $P'$ in $G'$ that corresponds to a weakly reachable path $\P=(\sigma,\theta)$ in
$\FSD(P,G)$, we can compute  $\phi(n,n)$ by propagating $\phi$-values across free space
intervals in a breadth-first manner. However,
we need to construct both $P'$ and $\P$ during our algorithm.
As opposed to generic map-matching algorithms for the weak Fr\'echet distance \cite{bpsw-mmvtd-05}, we need to ensure that $P'$ is simple.

Since $\phi$ is the length of a shortest path, it seems as if one could compute it by simply using a breadth-first propagation. However, one has to be careful because a weakly reachable path $\P$ is only allowed to backtrack {\em along the path in $G$ that it has already traversed}.
We therefore carefully combine two breadth-first propagations to compute the $\phi$ values for all $I\in{\cal I}$, where $\cal{I}$ is the set of all (non-empty) free space intervals on all spines $\Spine(v)$ for all $v\in V$.
For the primary breadth-first propagation, we initialize a queue $Q$ by enqueuing the interval $I\subseteq\spine(1)$ that contains $(1,1)$.
Once an interval has been enqueued it is considered {\em visited}, and it can never become unvisited again.
%
Then we repeatedly extract the next interval $I$ from $Q$. Assume $I\subseteq\spine(u)$. For each $v$ from $u+1$ to $n$ we consider $\Strip(u,v)$ and we compute all unvisited intervals $J\subseteq\spine(u)\cup\spine(v)$ that are reachable from $I$ with a path in $\strip(u,v)$. These $J$ can be reached using one more vertex, therefore we set $\phi(J)=\phi(I)+1$, we insert $J$ into $Q$, and we store the predecessor $\pi(J)=I$.
For each $J\in\spine(u)$ we then launch a secondary breadth-first traversal to propagate $\phi(J)$ to all unvisited intervals $J'\in{\cal I}$ that are reachable from $J$ within the free space of $\FSD(P,G(\pi(J)))$. Here, $G(\pi(J))$ denotes the projection of the predecessor DAG rooted at $\pi(J)$ onto $G$, i.e., each interval $I$ in the predecessor DAG is projected to $u$ if $I\subseteq\spine(u)$.
%
%
This allows $\P$ to backtrack along the path in $G$ that it has already traversed, {\em without} increasing $\phi$.
This secondary breadth-first traversal uses a separate queue $Q'$,
and sets $\phi(J')=\phi(J)$ and $\pi(J')=J$.
%
%
When this secondary traversal is finished, $Q'$ is prepended to $Q$, and then the primary breadth-first propagation continues.
Once $Q$ is empty, i.e., all intervals have been visited, $\phi(I)=\#\wFrechet_\V(P,\delta)$, where $I\subseteq\spine(n)$ is the interval that contains $(n,n)$. Backtracking a path from $n$ to $1$ in the predecessor DAG  $\pi(I)$ yields the simplified curve $P'$.
This algorithm visits each interval in ${\cal I}$ once using nested breadth-first traversals. Since there are $O(n^3)$ free space intervals this takes $O(n^3)$ time and space.

\begin{theorem}\label{thm:VRWFD}
	We can solve the vertex-restricted GCS problem under the weak Fr\'echet distance
	in $O(n^3)$ time and space.
\end{theorem}

\subsection{Fr\'echet Simplification $\Frechet_\V(P,\delta)$ in $O(n^3)$ Time}
\label{subsec:FrechetDP}

In this section we provide an algorithm to
solve the
vertex-restricted GCS problem under the Fr\'echet distance. 
%

%

\paragraph*{Elementary Intervals and Cost Function $\Phi$}
%
%
%
\begin{figure}[htbp]
	\centering
	\includegraphics[width=.5\textwidth]{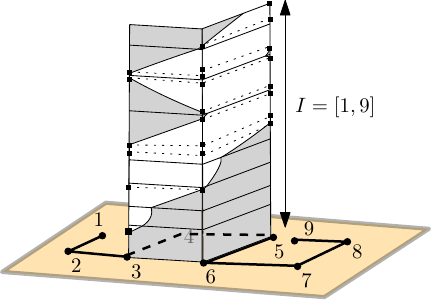}
	\caption{Elementary intervals on each spine are created by subdividing free space intervals with all endpoints of free space intervals on all other spines.}
	\label{fig:EI}
\end{figure}  
Let $S\subseteq [1,n]$ be the set of all interval endpoints of
all free space intervals of $\spine(v)\subseteq \Spine(v)=[1,n]\times v$
for all $v\in V$, projected onto
$[1,n]$. The set $S$ induces a partition
of $[1,n]$ into intervals;
these are half-open $[a,b)$, except for the last interval $[a,b]$
which is closed. For each $v\in V$, let $\mathemph{\slist(v)}$ be the ordered
list of \myemph{elementary intervals} obtained by subdividing the
intervals of $\spine(v)$ according to this partition; see
Figrue~\ref{fig:EI}.
%
We assume that elementary intervals in $\slist(v)$ are ordered in increasing order of their starting point. Projecting all $\slist(v)$ to $[1,n]$ induces a total order on all elementary intervals; we use $<$ and $=$ to compare intervals.



We define the cost function $\phi: [1,n] \times V\rightarrow \Nats$ as
$\mathemph{\phi(z,v)}=\min_{\P} \# \P$, where the minimum ranges over
all reachable paths from $(1,1)$ to $(z,v)$ in the free space
diagram. Here, $\P$ is called \myemph{reachable} if it starts in $(1,1)$, lies entirely
in free space, and is monotone in $P$ and on each edge in $E$.
If no such path exists then $\phi(z,v)=\infty$.
The function $\phi$ is defined on all spines $\spine(v)$ for all $v\in
V$, and it captures min-link reachability in the free space diagram.
We will propagate $\phi$ across
the free space diagram using dynamic programming.
%

\begin{lemma}[Properties of elementary intervals] \label{lem:EIProp}
	Let $v\in V$. Then:\\
	\begin{enumerate*}
		\item $|\slist(v)|\leq 2n^2+1$, and \label{prop:size}
		\item $\phi(a,v) = \phi(b,v)$ for all $a,b\in e\in \slist(v)$.  \label{prop:equalcost}
		%
	\end{enumerate*}
\end{lemma}

\begin{proof} 
	\ref{prop:size}.  There are at most $n$ free space intervals on each
	spine, and there are $n$ spines, therefore there are $|S|=2n^2$
	interval endpoints. These subdivide $[1,n]$ into $2n^2+1$ intervals.

	\ref{prop:equalcost}. For the sake of contradiction, assume there
	exist $a,b\in e$, with $a\neq b$ such that $\phi(a,v) \neq \phi(b,v)$. Assume without loss of generality that
	$\phi(a,v)<\phi(b,v)$. Let $\P(a,v)$ be a reachable path to $(a,v)$ with
	$\#\P(a,v)= \phi(a,v)$.  (i) If $b>a$, then extending $\P(a,v)$ to
	continue vertically from $a$ to $b$ along $e$ yields a reachable path
	$\P(b,v)$ to $(b,v)$ with $\#\P(b,v)=\#\P(a,v)$; see
	\figref{fig:propProof}(a). Therefore $\phi(b,v)\leq\phi(a,v)$ which is
	a contradiction.
	
	(ii) Now, assume that $a>b$.  Let $s$ be the start point of $e$, and
	let $p$ be the last point on $\P(a,v)$ such that $p=(j,s)$ for some
	$j$.  Then $p$ lies in some strip $\strip(u,v)$.  Since $p,(a,v),(b,v)$ all
	lie in free space and $e$ is a subset of a free space interval of a
	single cell, $p,(a,v),(b,v)$ all lie in a single free space cell.  Hence we
	can extend $\P(p)$ with a line segment to $(b,v)$, which yields a
	reachable path $\P(b,v)$ with $\#\P(b,v)=\#\P(a,v)$; see
	\figref{fig:propProof}(b). Therefore $\phi(b,v)\leq\phi(a,v)$ which is
	a contradiction.
	\begin{figure}
		\centering
		\includegraphics[width=.6\linewidth]{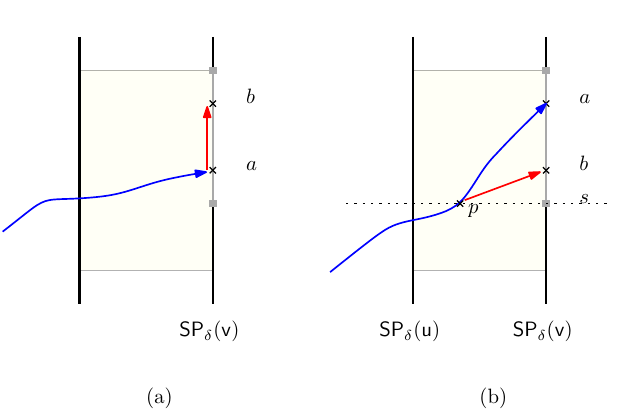}
		\caption{Illustration of the proof of property \ref{prop:equalcost}. $a$ and $b$ belong to the same elementary interval $e$ that is highlighted in gray. (a) $b>a$ (b) if $a>b$.}
		\label{fig:propProof}
	\end{figure}
\end{proof}

By \lemref{lem:EIProp}, $\phi$ is constant on each elementary interval $e \in \slist(v)$.
For brevity we write $\phi(e,v)$ and we notate $\P(e,v)$ as a reachable path to $e$.
\lemref{lem:DPFormula} below shows the recursive formula for $\phi$ which we will use in our dynamic programming algorithm.

\begin{lemma}[Recursive formula]
	\label{lem:DPFormula}
	$\;$\vspace*{-1.5ex}  
	\begin{enumerate}
		\item For all $e\in \slist(1)$: If $e$ is
		reachable then $\phi(e,1)=0$, otherwise $\phi(e,1)=\infty$.
		\item For all $v\in\{2,\ldots,n\}$ and $e\in \slist(v)$:
		$\phi(e,v) = \min\limits_{1\leq u< v}\,\min\limits_{\substack{e'\in \slist(u)\\ e'\leq e}}\phi(e',u) + 1$ ,\\ where the
		second minimum is taken over all $e'$ such that
		there exists a reachable path from $e'$ to $e$ within $\strip(u,v)$.
	\end{enumerate}
\end{lemma}
\begin{proof} 
	(1) If $e$ is reachable, then we know that $\phi(e,1)=\#\P(e,1)=0$ because $e$ belongs to the first spine $\slist(1)$. If $e$ there is no reachable path to $e$ then $\phi(e,1)=\infty$.
	
	(2) By definition and Lemma~\ref{lem:EIProp}, $\phi(e,v)=\min_{\P(e,v)} \# \P(e,v)$, where the minimum ranges over all reachable paths to $(e,v)$ in the free space diagram, and reachable paths are monotone in $P$ and on each edge in $E$.
	
	Let $\P^*(e,v)$ be a reachable path to $(e,v)$ such that $\#\P^*(e,v)=\phi(e,v)$.
	Now consider the last spine $\Spine(u)$ that $\P^*(e,v)$ visits before
	visiting $\Spine(v)$, and let $e'\in\slist(u)$ be the last elementary
	interval it visits.  Let $\P_1$ be the sub-path of $\P^*(e,v)$ that
	ends in $e'$ on $\Spine(u)$, and let $\P_2$ be the remaining portion
	of $\P^*(e,v)$ that starts in $e'$ on $\Spine(u)$ and ends in $e$ on
	$\Spine(v)$. Then $\P_2$ is a reachable path from $(e',u)$ to $(e,v)$,
	and due to monotonicity in $E$, $\P_2$ has to lie in $\Strip(u,v)$. Therefore
	$\#\P_2=1$ and $\#\P^*(v,e) = \#\P_1+\#\P_2=\#\P_1+1$ by
	construction. And we know that $\#\P_1$ has to be the minimum number
	of links to reach $(e',u)$, because otherwise $\# \P(e,v)$ would not
	be optimal. Hence, $\phi(e',u)=\#\P_1$ and we have
	$\phi(e,v)=\phi(e',u)+1$.

	Since edges in $E$ are directed, any reachable path to $(e,v)$ can only visit spines $\Spine(u)$ for $u\in \{1,\ldots,v\}$. And any such path has to be monotone in $P$ and therefore can only visit elementary intervals $e'<e$.
	Thus,  $\phi(e,v) = \min\limits_{1\leq u< v}\,\min\limits_{\substack{e'\in \slist(u)\\ e'\leq e}}\phi(e',u) + 1$,
	where only those $e'$ are considered for which there is a reachable path from $e'$ to $e$  in $\Strip(u,v)$.
\end{proof} 

\paragraph*{DP Algorithm}

%
First we compute the shortcut DAG
$G$ and the free space diagram $\FSD(P, G)$. This can be done
by computing free space diagrams for all strips (and spines) and
connecting them together with respect to the adjacency information in
$G$.  For every $v\in V$ we then compute the list $\slist(v)$ of all
elementary intervals in $\Spine(v)$, and we initialize $\phi(v,e)$
according to Lemma~\ref{lem:DPFormula}.
%
%
The algorithm
processes the free space diagram spine by spine
for $v\in\{2,\ldots,n\}$.

For fixed $v$ we
compute $\phi(e,v)$ for all $v\in\slist(v)$
by propagating reachability
according to the recursive formula in
Lemma~\ref{lem:DPFormula} using the following batching
approach; refer to Figure~\ref{fig:Subdivision}
for an illustration:
%
%
%
\begin{figure}[htbp]
	\centering
	\includegraphics[width=.9\textwidth]{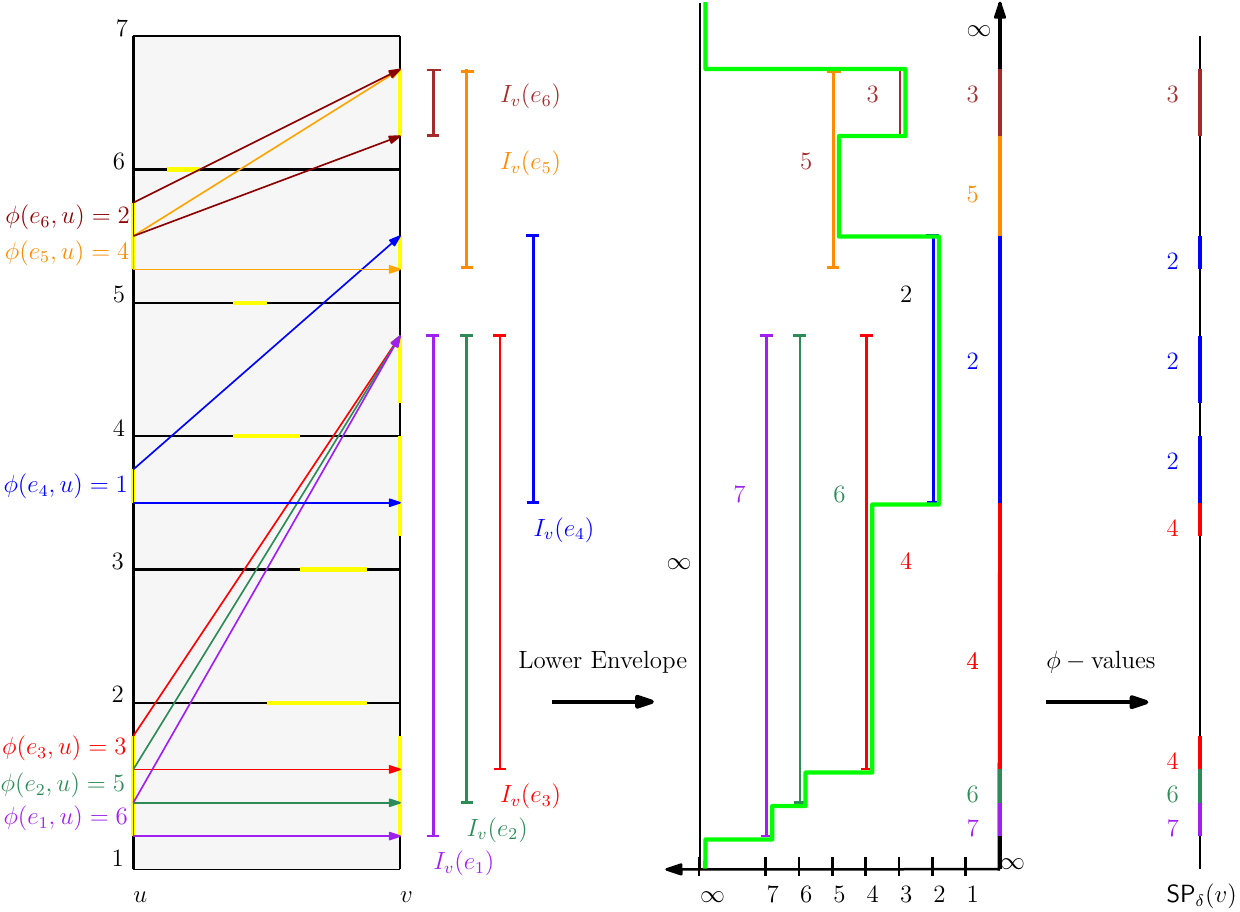}
	\caption{Propagating $\phi$ within $\strip(u,v)$.
		%
		(a) shows reachable intervals on $\Spine(v)$ computed for elementary intervals in $\slist(u)$.
		(b) shows the computation of the lower envelope $\L_u$ of all all reachable
		intervals $I_v(e')$ for all $e'\in\slist(u)$.
		%
		(c) shows the resulting $\phi$-values.}
	\label{fig:Subdivision}
\end{figure}
For all $u<v$ and all elementary intervals
$e'\in\slist(u)$ we compute the
\myemph{reachable interval} $\mathemph{I_v(e')}\subseteq\Spine(v)$
which is the smallest (continuous) interval in $\Spine(v)$ that includes all points
$z\in\Spine_\delta(v)$ that are reachable from $e'$ within $\strip(u,v)$. Note that $I_v(e')$
generally includes points from multiple elementary intervals as well as
non-reachable points. For all
points $z\in I_v(e')$ that can be reached with a path that visits $e'$ and
traverses $\strip(u,v)$,
we know that $\phi(z,v)\geq \phi(e',u)+1$. We therefore associate the
cost $\mathemph{\psi(I_v(e'))}=\phi(e',u)+1$ with the reachable
interval $I_v(e')$.
Thus, $\psi(I_v(e'))$ is a partially defined constant function on
$\Spine(v)$, the graph of which is a line segment $I_v(e')\times\psi(I_v(e'))$.
For each $u<v$, we compute the
\myemph{lower envelope} $\mathemph{\L_u}:\Spine(v)\rightarrow\Reals$ of all line segments $I_v(e')\times\psi(I_v(e'))$ for
all $e'\in\slist(u)$, which is the  pointwise minimum of all
these partially defined constant functions. 
%
%
%
%
%
For each elementary interval $e\in\slist(v)$, we have that $\L_u$ is constant
by construction, and we compute $\phi(e,v) = \min_{1\leq u<v}\L_u(e)$,
which follows from the recursive formula in Lemma~\ref{lem:DPFormula}. 
%
After having processed all $v\in \{2,\ldots,n\}$, all $\phi(e,v)$ have been computed.
We then construct a reachable path $\P((1,1),(n,n))$ by backtracking the $\phi$-values from $\phi(n,n)$.
The sequence of traversed strips corresponds to the desired
path $P'$ in $G$.
%

There are a total of $O(n^3)$ elementary intervals.
%
%
For fixed $u<v$, 
using Lemma 3 of~\cite{aerw-mpm-03} we can compute all reachable intervals
within a single strip $\Strip(u,v)$ in $O(|\slist(u)|)=O(n^2)$ time.
Each lower envelope $\L_u$ can be
computed in time linear in the number of intervals, hence $O(n^2)$ time.
Computing $\phi(e,v) = \min_{1\leq u<v}\L_u(e)$ for all  $e\in\slist(v)$ takes $O(n^3)$ time,
thus the total runtime is $O(n^4)$ using
$O(n^3)$ space.
%

\begin{theorem}\label{thm:VRFD}
	Given a polygonal curve $P$ with $n$ vertices and $\delta>0$, an optimal solution to the vertex-restricted GCS problem under Fr\'echet distance can be computed in $O(n^4)$ time and $O(n^3)$ space.
\end{theorem}

\begin{algorithm}[htbp]
	\DontPrintSemicolon
	\SetKwFunction{MinLinkSimp}{\textsc{MinLinkSimp{$(P,\delta$)}}}
	\SetKwFunction{Subdivide}{\textsc{Subdivide{}}}
	\SetKwFunction{RIf}{\textsc{ReachInterval{}}}
	\SetKwComment{Comment}{}{}
	\SetKwInOut{Input}{Input}\SetKwInOut{Output}{Output}
	\SetKwFor{ForAll}{for all}{:}{}
	\SetKwFor{For}{for}{:}{}
	\caption{Algorithm for vertex-restricted GCS problem for $\Frechet_\V(P,\delta)$}
	\label{alg:VRFD}
	
	\Input{Polygonal curve $P=\langle p_1,p_2,\cdots, p_n \rangle$, and $\delta>0$}
	\Output{Optimal vertex-restricted simplification $P'$ of $P$ such that $\Frechet(P,P')\leq\delta$}
	
	Compute shortcut DAG $G=(V=\{1,\ldots,n\},E)$\;\label{step:DAG}
	Compute $\FSD(P, G)$\;\label{step:FSD} 
	\ForAll{$v\in V$}{
		Compute the list of elementary intervals $\slist(V)$\;
		\lForAll{$e\in \slist(v)$}{$\phi(v,e)=\infty$\label{step:intzinf}
			\Comment*[h]{\rm\em\color{lgray} // Initialize $\phi$}}}
	
	\lForAll{$e \in \slist(1)$ that is reachable from $(1,1)$}{$\phi(1,e)=0$ \label{step:intzzero}
		\Comment*[h]{\rm\em\color{lgray} // Initialize $\phi$}}
	
	\For{$v=2$ \KwTo{$n$} \label{step:prpg}}{
		\ForAll{$u\in \{1,\cdots, v-1\}$}{ \label{step:allU}
			$S = \emptyset$ \label{step:empty}\;  
			\ForAll{$e'\in \slist(u)$}{ \label{step:iterateElementaryStart}
				Compute the reachable interval $I_v(e')$ and set $\psi(I_v(e')) = \phi(e',u)+1$\;
				$S = S\cup\{I_v(e')\}$ \label{step:union_RI}
			}
			Compute the lower envelope $\L_u$ of all $I\times \psi(I)$ for all $I\in S$
		}
		
		\lForAll{$e\in\slist(v)$}{$\phi(e,v) = \min_{1\leq u<v}\L_u(e)$\label{step:lowerenvAssign}}
	}

	Return vertices of $P'$ by tracing back the $\phi$ values.   \label{step:completion}
\end{algorithm}

\begin{proof} 
	The main task is to compute $\phi(n,n)$. The shortcut DAG can be
	computed straight-forwardly in $O(n^2)$ time and space, and the free space diagram
	requires computing all spines and strips. There are $O(n^2)$ strips
	and it takes $O(n)$ time and space to compute the free space diagram
	for each strip, for a total of $O(n^3)$ time and $O(n^2)$ space. By
	Lemma~\ref{lem:EIProp} there are $O(n^2)$ elementary intervals per
	spine and thus $O(n^3)$ elementary intervals total, and thus all the
	$\slist(v)$ can be computed in $O(n^3)$ time.
	The $\phi$-value of each of the $O(n^3)$ elementary intervals
	is initialized in lines \ref{step:intzinf} and \ref{step:intzzero} in overall $O(n^3)$ time.
	
	
	Lines~\ref{step:prpg}-\ref{step:lowerenvAssign} compute all
	$\phi(e,v)$ for all $e\in\slist(v)$ by batching the reachability
	propagations according to the recursive formula in
	Lemma~\ref{lem:DPFormula}. Lemma 3 of~\cite{aerw-mpm-03} shows that
	within a single strip $\Strip(u,v)$, all reachable intervals for all
	$O(n)$ {\em free space intervals} can be computed in $O(n)$ time. In
	our case, for fixed $u<v$, we have $O(n^2)$ elementary intervals in
	$\slist(u)$ (which are subdivisions of the free space intervals), and
	we can apply this lemma to compute all their reachable intervals in
	$O(n^2)$ time.
	We assign to each reachable interval $I_v(e)$ the cost $\psi$ that it
	takes to reach $e$ with a min-link path through $e'$.
	For fixed $v$ and $u$, lines
	\ref{step:iterateElementaryStart}-\ref{step:union_RI} take $O(n^2)$ time.
	Processing elementary intervals $e\in\slist(u)$ in
	increasing order, yields reachable intervals in increasing
	order of their start points. This allows us to compute the lower envelope $\L_u:\spine(v)\rightarrow\Reals$
	in linear time in the number of intervals (see Section 5
	of~\cite{aaghi-vdp-85}), hence in $O(n^2)$ time.
	By construction, $\L_u$ is constant on each  $e\in\slist(v)$ . 
	From the recursive formula in Lemma~\ref{lem:DPFormula}, 
	follows that $\phi(e,v) = \min_{1\leq u<v}\L_u(e)$, which is computed in $O(n^3)$ time in line \ref{step:lowerenvAssign}.
	%
	Therefore the entire runtime is $O(n^2\cdot n^2+n\cdot n^3)=O(n^4)$, and
	the overall space needed is $O(n^3)$.
	A min-link reachable path in the free space diagram can be traced back from $\phi(n,n)$ in $O(n^2)$ time, and
	the sequences of spines visited correponds to the simplification $P'$.
\end{proof}

\paragraph*{Improved DP Algorithm}
We now describe an improvement of the algorithm to $O(n^3)$ time.
%
The runtime bottleneck of the algorithm is that, for fixed $v\in V$, the propagation step
which computes all $\phi(e,v)$ for all $e\in\slist(v)$, takes $O(n^3)$ time, even though there are only
$|\slist(v)|=O(n^2)$ elementary intervals.
It turns out that the $\psi$-values of the reachable intervals have a special structure that
we can exploit to speed up this propagation to take only $O(n^2)$ time.

For this we need to consider portions of strips and spines:
Let $u<v$.
We define $\mathemph{\Strip^{[i,j]}(u,v)}=[u,v]\times[i,j]$ and $\mathemph{\Spine^{[i,j]}(u)}=u\times[i,j]$.
And let $\mathemph{\slist^{[i,j]}(u)}=\slist(u)\cap\Spine^{[i,j]}(u)$.
%
For each elementary interval $e\in\slist^{[i,j]}(u)$, we break the reachable
interval $I_v(e)\subseteq\Spine(v)$ into the \myemph{upper reachable interval}
$\mathemph{I^\top_v(e)}=I_v(e)\cap \Spine^{[i+1,n]}(v)$
and the \myemph{lower reachable interval}
$\mathemph{I^\bot_v(e)}=I_v(e)\cap \Spine^{[i,i+1]}(v)$.
%

\begin{observation}
	\label{obs:identicalTop}
	Let $e,e'\in \slist^{[i,i+1]}(u)$. Then:
	\vspace*{-1ex}
	\begin{itemize}
		\item $I^\top_v(e)$ and $I^\top_v(e')$ are identical.
		\item $I^\bot_v(e)=[a,i]$, where $a$ is the maximum of the bottom endpoints
		of $e$ and the free space interval on $\spine^{[i,i+1]}(u)$, or
		$\infty$ if the free space interval is empty.
	\end{itemize}
\end{observation}
This means that all elementary intervals in $\slist^{[i,i+1]}(u)$  generate the same upper reachable interval, and lower reachable intervals depend on bottom endpoints only.
%
%
\begin{lemma} [Monotonicity of $\Phi$] \label{lem:decreasingValues}
	Let $e_1<\cdots < e_m$ be all elementary
	intervals in $\slist^{[i,i+1]}(u)$.
	Then $\phi(e_m,u)\leq \ldots \leq \phi(e_2,u)\leq \phi(e_1,u)$.
\end{lemma}
\begin{proof}
	We use proof by contradiction. Let $j$ and $k$ be two integers such that $1\leq j<k\leq m$, $e_j < e_k$ and $\phi(e_k,v) > \phi(e_j,v)$. Then a reachable path $\P(e_j,1)$ of length $\#\P(e_j,1)=\phi(j,v)$ is also reaching $e_k$ since all $e_1,\cdots , e_m$ are within free space interval of $[i,i+1]\cap\spine(v)$. Therefore, $\#\P(e_k,1)= \#\P(e_j,1) = \phi(e_j,v)$ and we have $\phi(e_k,v)\leq \phi(e_j,v)$ which is a contradiction.
\end{proof}

%
%
%
Now let $\mathemph{\mu^{[i,i+1]}(v)}$
be the minimum $\phi$-value
in $\slist^{[i,i+1]}(v)$.
And for each $e\in\slist^{[i,i+1]}(v)$ let
\[\mathemph{\bar{\phi}(e,v)}=\min_{1\leq u<v}\min_{\substack{e'\in \slist^{[i,i+1]}(u)\\ e'\leq e}} \phi(e',u) + 1
= \min_{\substack{e'\in \slist^{[i,i+1]}(v-1)\\ e'\leq e}} \bar{\phi}(e',v-1) + 1 \;,\]
where only those $e'$ are considered for
which there exists a reachable path from $e'$ to $e$ within
$\strip^{[i,i+1]}(u,v)$.
The function $\bar{\phi}$ propagates horizontal reachability within $[i,i+1]$ only.
%

%
The propagation now proceeds in two steps, for fixed $v$:
\begin{enumerate}
	\item For all $u<v$ and all $1\leq i<n$ we compute the upper interval $I^\top_v(e')$
	for any fixed $e'\in\slist^{[i,j]}(u)$ using Lemma 3 of~\cite{aerw-mpm-03}, and we know that $\psi(I^\top_v(e'))=\mu^{[i,i+1]}(u)+1$.
	This generates $O(n^2)$ upper reachable intervals. We compute
	their lower envelope $\L^1$.
	\item We propagate lower reachable intervals \myemph{only from $u=v-1$ using $\bar{\phi}$}:
	For all $e'\in\slist(v-1)$ we compute the lower interval $I^\bot_v(e')$. This can be done in constant time per interval, and it takes $O(|\slist(v-1)|)=O(n^2)$ time. We set  $\psi(I^\top_v(e'))=\bar{\phi}(e')+1$. We compute their lower envelope $\L^2$.
	Finally we need to update $\phi$ using the lower envelope of $\L^1$ and $\L^2$ and we also update $\bar{\phi}$ using $\L^2$ only.
\end{enumerate}
Thus the runtime for an update step is $O(n^2)$ and the total runtime is $O(n^3)$.

\begin{theorem}\label{thm:VRFD_improved}
	Given a polygonal curve $P$ with $n$ vertices and $\delta>0$, an optimal solution to the vertex-restricted GCS problem under Fr\'echet distance can be computed in $O(n^3)$ time and $O(n^3)$ space.
\end{theorem}

We note that we can save space by only storing $\mu^{[i,i+1]}(v)$ for all $v\in V$
and all $1\leq i\leq n$, and storing $\phi$ and $\bar{\phi}$ on strips
$\Strip(v-1)$ and $\Strip(v)$ only, for a total of $O(n^2)$ space. In
order to trace back a reachable path, however, one then needs to apply
a variant of Hirschberg's trick, which results in a runtime of $O(n^3\log n)$ space.

\section{Vertex-Restricted GCS under the Hausdorff Distance from $P'$ to $P$} \label{sec:VRHausP'toP}
In this section we revisit the GCS problem for $\diHausdorff_{\V}(P,\delta)$ considered by Kreveld \etal~\cite{klw-oopsihfd-18}. We improve on the running time of their $O(n^4)$ time algorithm. 
First we thicken the input curve $P$ by width $\delta$. This induces a polygon $\P$ with $h=O(n^2)$ holes. Now all we need is to decide whether each shortcut $\segment{p_ip_j}$ for all $1\leq i< j\leq n$ lies entirely within $\P$ or not. To this end, we proprocess $\P$ into a data structure such that for any straight line query ray $\rho$ originated from some point inside the $\P$, compute the first point on the boundary of $\P$ hit by $\rho$. In other words we have a collection of $h$ simple polygons of total complexity of $N=O(n^2)$. We use the data structure proposed by~\cite{cegghss-rspugt-94} of size $O(N)$ which can be constructed in time $O(N\sqrt{h}+h^{3/2}\log h+ N\log N)$ and which answers queries in $O(\sqrt{h}\log N)$ time. We have $\Theta(n^2)$ shortcuts $\segment{p_ip_j}$ to process and need to examine whether each shortcut lies inside $\P$ or not i.e., whether there is a simple polygon (hole in $P$) among that is hit by $\rho$. If a shortcut lies inside $\P$ then we store it into the edge set of the shortcut graph proposed by Imai-Iri~\cite{ii-aoaaplf-86}. Otherwise we eliminate the shortcut. We originate a ray at $p_i$ and compute the first point $x$ on the boundary of $\P$ hit by the ray in $O(\sqrt{h}\log N)$ query time. All we need is to compare the length of the ray to the length of the shortcut. If $\|p_i-x\| \geq \|p_i-p_j\|$ then the shortcut lies inside $\P$, otherwise it does not. Once the edge set of the shortcut graph is constructed, we compute the shortest path in it. As a result we have the following theorem:

\begin{theorem} \label{thm:VRHP'toP}
	Given a polygonal curve $P$ with $n$ vertices and $\delta>0$, an optimal solution to the vertex-restricted GCS problem for ${\diHausdorff}_{\V}(P,\delta)$ can be computed in $O(n^3\log n)$ time and $O(n^2)$ space.
\end{theorem}

\section{NP-Hardness Template for Curve-restricted GCS} \label{sec:CR}

In this section we construct a template that we use to prove NP-hardness of the curve-restricted GCS problems for most of the distance measures discussed in this paper. The template takes inspiration from the NP-hardness proofs of minimum-link path problems~\cite{klps-ocmlpp-17}. We believe that this template can be adapted to show hardness of other similar problems.

The template reduces from the \textsc{Subset Sum} problem.
Given a set of $m$ integers $A=\{a_{1},a_{2},\dots,a_{m}\}$ and an integer $M$, we will construct an instance of the curve-restricted GCS problem such that there exists a subset $B\subset A$ with the total sum of its integers equal to $M$ if and only if there exists a simplified polygonal curve with at most $2m+1$ vertices.

\subsection{Overview}
The input curve $P$ we construct has a zig-zag pattern.  It has $m$ \emph{split gadgets} at every other bend of the pattern, $m+1$ \emph{enumeration gadgets} at the other bends, and $2m$ \emph{pinhole gadgets} halfway through each zig-zag segment (refer to Figure~\ref{fig:sketch}). The split and the enumeration gadgets are the same for all the distance measures, and only the pinhole gadgets vary. 
The construction forces any optimal simplification \(P'\) to follow a zig-zag pattern. The pinhole gadget is named as such because any segment of \(P'\) that goes through it is forced to pass through a specific point, called the \emph{pinhole}. This limits the placements of \(P'\)'s vertices. 
The choice of where to place the vertex on each split gadget corresponds to the choice of including or excluding a given integer in the subset \(B\). The $x$-coordinate of the vertex of \(P'\) on each enumeration gadget encodes the sum of integers in \(B\) up to that point. We will ensure that the last point of \(P\) is reachable using at most \(2m+1\) vertices only if \(B\) sums to exactly \(M\).
Different distance measures require different pinhole gadgets to make this construction work.
For the reduction to be successful the following properties must hold for the construction:

\begin{enumerate}
	\item Any segment of \(P'\) starting before a pinhole gadget and ending after the pinhole gadget must pass through the pinhole gadget's \emph{pinhole}.
	\item It must be impossible to have a segment of \(P'\) traverse multiple pinhole gadgets at once.
	\item Any segment of \(P'\) where the starting vertex \(u\) is on a split or enumeration gadget, the segment goes through a pinhole, and the ending vertex \(v\) is on the next enumeration or split gadget, must have distance \(\leq \delta\) to \(P[u,v]\).
	\item \(P\) must be polynomial in size. Specifically, only a polynomial number of polyline segments can be used and all vertices must have rational coordinates.
\end{enumerate}
In Section~\ref{sec:proofofproperties} we will show that a construction with these properties implies NP-hardness. 

\begin{figure}[t]
	\centering
	\includegraphics{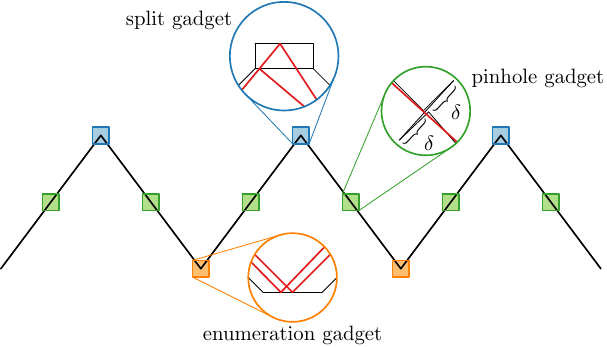}
	\caption{Sketch of our template curve}
	\label{fig:sketch}
\end{figure}

\subsection{Exact construction}
Exact coordinates for the construction are based on two constants \(\delta\) and \(\gamma\). The value \(\delta\) is the bound on the allowed distance between a simplification and the original curve, and \(\gamma\) is a constant that is significantly larger than \(\delta\). For our construction, we set \(\delta\) to be \(4\sum a_i\), and \(\gamma\) to be \(\delta^2\).

\subparagraph{Split gadget} The split gadgets, denoted by $\sigma^i$ (for $1\le i \le m$), consist of a chain of five segments (refer to Figure~\ref{fig:splitgadget}). Their vertex coordinates are as follows (with vertex index denoted by subscript):
\[
\begin{aligned}
\sigma^i_1,\sigma^i_5 & = \left(\frac{3(2i-1)}{4}\gamma,\gamma \right)\,,\\
\sigma^i_2,\sigma^i_6 & = \left(\frac{3(2i-1)}{4}\gamma+w, \gamma \right)\,,\\
\sigma^i_3 & = \left(\frac{3(2i-1)}{4}\gamma+w, \gamma+h_i \right)\,,\\
\sigma^i_4 & = \left(\frac{3(2i-1)}{4}\gamma, \gamma+h_i \right)\,,
\end{aligned}
\]
where \(h_i=\frac{2\gamma a_i}{3\gamma-4a_i}\) for \(i \in \{1,\dots,m\}\), and \(w=\delta/2\) is the width of the split gadgets. 
The segments \(P[\sigma^i_1,\sigma^i_2]\) and \(P[\sigma^i_3,\sigma^i_4]\) are called \(\sigma^i\)s \emph{lower- and upper-mirror segments} respectively.

\subparagraph{Enumeration gadget} The enumeration gadgets, denoted by $\mu^i$ (for $1\le i \le m-1$), consist of just one horizontal segment, with the following vertices:
\[
\begin{aligned}
\mu^i_1 & = \left(\frac{3i}{2}\gamma-w, 0 \right)\,,\\
\mu^i_2 & = \left(\frac{3i}{2}\gamma, 0 \right)\,.\\
\end{aligned}
\]
There are two special enumeration gadgets: the first enumeration gadget \(\mu^0\) consists of one point with coordinates \((0,0)\) (the starting point of the curve $P$), and the last enumeration gadget \(\mu^m\), also called the \emph{budget segment}, consists of two vertices with the following coordinates:
\[
\begin{aligned}
\mu^m_1 & = \left(\frac{3m}{2}\gamma-w, 0 \right)\,,\\
\mu^m_2 & = \left(\frac{3m}{2}\gamma-M, 0 \right)\,.
\end{aligned}
\]

\subparagraph{Pinhole gadget} The exact structure of the pinhole gadgets depends on the curve distance measure. It is important for us however to specify the coordinates of its \emph{pinhole}, that is, the fixed point that every simplification segment must pass through. We denote these pinholes as $c^i$ for $0 \le i \le 2m-1$. Their coordinates are:
\[
c^i = \left(\frac{3}{8}\gamma+\frac{3i}{4}\gamma,\frac{1}{2}\gamma\right)\,.
\]
Each segment of a simplified polygonal curve directed from an enumeration to a split gadget has to pass through a pinhole with an even index, and each segment from a split to an enumeration gadget has to pass through a pinhole with an odd index.

\subsection{Proof of the construction}
\label{sec:proofofproperties}
As stated in the overview, there are several properties that must be satisfied for the construction to work. We will begin by showing that as long as a particular pinhole gadget inserted in the template does not significantly alter the structure of the curve, the last three properties are already satisfied by our construction no matter the shape of the pinhole gadget.

Property 2 holds because all of the pinholes have the same $y$-coordinate. So any segment traversing two pinholes must be a horizontal one with $y$-coordinate \(\gamma/2\), which makes the distance between the segment and the enumeration or split gadget between the pinholes greater than \(\delta\). 
Property 4 obviously holds as our curve has a polynomial number of vertices with rational coordinates.
For Property 3 we have the following lemma: 

\begin{lemma}
	\label{lem:splitenumeratecovered}
	The split and enumeration gadgets can be covered by a vertex of a simplification on the gadget under distances $\overrightarrow{\mathsf{H}}(P,\delta)$, $\overleftarrow{\mathsf{H}}(P,\delta)$, $\mathsf{H}(P,\delta)$, $\mathsf{F}(P,\delta)$, $\mathsf{wF}(P,\delta)$.
\end{lemma}

\begin{proof}
	The width of an enumeration segment is $\delta/2$. Trivially, under all the distance measures, a point on it will cover the whole segment.
	
	The width of a split gadget $\sigma^{i}$ is $\delta/2$ as well, and the height is less than $a_{i}$ which in turn is less than $\delta/2$. Then, the diagonal of the rectangle forming the gadget is less than $\delta/\sqrt{2}$ (refer for an example to Figure~\ref{fig:splitgadget}). Thus, any point on the split gadget is within distance $\delta$ of all the points of the gadget under any of the considered distance measures.
	\begin{figure}
		\centering
		\includegraphics{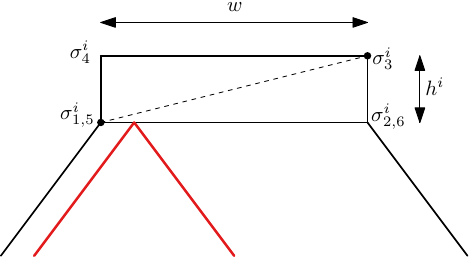}
		\caption{Split gadget shown in black, with a simplification in red. The diagonal distance (dashed) is less than $\delta/\sqrt{2}$, thus all the points of the split gadget are within distance $\delta$ to the vertex of the simplification curve. 
		}
		\label{fig:splitgadget}
	\end{figure}
\end{proof}

Since split and enumeration gadgets are completely covered by any points we place on them, property 3 of our construction holds as long as our pinhole gadget has distance at most \(\delta\) to any line segment between a split and enumeration gadget that passes through the pinhole.

Now that we have shown our properties to hold (given a correct pinhole gadget), we will show how the construction enforces that any optimal simplification of \(P\) with at most \(2m+1\) vertices encodes a subset of \(A\) that sums to exactly \(M\).
\begin{lemma}
	\label{lem:unigadgets}
	If the input curve \(P\) is constructed with pinhole gadgets satisfying the four template properties, any optimal simplification $P'$ of $P$ must have a vertex on all split and enumeration gadgets in the corresponding order, passing through the pinhole of each pinhole gadget.
\end{lemma}
\begin{proof}
	Note that we count the split gadgets starting from $1$, and the enumeration segments starting from $0$.
	We will prove that every segment of a valid simplification of size at most $2i$ of the prefix of $P$ up to enumeration segment \(\mu^i\) passes through a pinhole, and that each segment connects a split and an enumeration gadget.
	There are \(2i\) pinhole gadgets between \(\mu^0\) and \(\mu^i\).
	Since \(P\), by assumption, satisfies Property 2 (no simplification segment can traverse multiple pinholes at once), and Property 1 (any segment of a simplification starting before and ending after a pinhole gadget must pass through the pinhole), at least \(2i+1\) vertices are needed to reach \(\mu_i\) from $\mu_{0}$. Property 3 implies that a simplification with a vertex on each enumeration and split gadget is valid and since it has \(2i+1\) vertices it is optimal. Finding a valid simplification that does not have a vertex on each split and enumeration gadget requires additional vertices so that the simplification is still within distance \(\delta\) of each gadget. This would imply that such a simplification is not optimal.
	%
\end{proof}

\begin{lemma}\label{lem:encodesubsetsum}
	The set of points on the $i$th enumeration gadget that a valid simplification curve with at most $2i+1$ vertices can reach encode all possible subsets of \(\{a_1,\dots, a_i\}\) in their $x$-coordinates: For a given subset \(B'\), the horizontal distance from the corresponding point to the right end of the enumeration segment is $\sum\limits_{a\in B'}a$.
\end{lemma}
\begin{proof}
	From Lemma~\ref{lem:unigadgets} we know that a simplification that reaches \(\mu^i\) using at most \(2i+1\) vertices has each segment pass through a pinhole. Knowing this, we can prove the lemma by induction.
	
	Consider the base case \(i=1\). If we draw a line from \(\mu^0\) through pinhole \(c^0\) there are two intersection points with \(\sigma^1\) which are thus the possible locations of the first vertex of an optimal simplification. The coordinates of these reachable points are \((\frac{3}{4}\gamma,\gamma)\) on the lower mirror segment and \((\frac{3}{4}\gamma+\frac{3}{4}h_1,\gamma+h_1)\) on the upper mirror segment. Drawing lines from these points through the next pinhole $c^{1}$ gives two possible intersection points with \(\mu^1\): \((\frac{3}{2}\gamma,0)\), for the line from the point on the lower mirror segment, and \((\frac{3}{2}\gamma-a_1,0)\), for the line from the point on the upper mirror segment.
	The $x$-coordinates of these points are of the form $\frac{3}{2}\gamma-S$, where $S$ is the sum of the elements in any possible subset of $\{a_{1}\}$, namely the empty set and the set $\{a_{1}\}$ itself. 

	Now for the general case, assume that any valid simplification with $2(i-1) + 1$ vertices can reach a set of points on \(\mu^{i-1}\) that encode all of the subsets of \(\{a_1,\dots,a_{i-1}\}\). That is, the $x$-coordinate of a reachable point is of the form \(\frac{3(i-1)}{2} - S\), where \(S\) is the sum of the elements in some subset of the first $i-1$ integers in $A$. We will now show that going from any one of these reachable points through the next two pinholes will allow us to reach precisely the points \((\frac{3i}{2} - S,0)\) and \((\frac{3i}{2}-S-a_i,0)\) on \(\mu^i\), corresponding to the two subsets created by either including \(a_i\) into that subset or not. 
	
	The line from \((\frac{3(i-1)}{2}\gamma - S,0)\) through \(c^{2(i-1)}\) (the pinhole between \(\mu^{i-1}\) and \(\sigma^i\)), intersects \(\sigma^i\) in the point \((S-\frac{3}{4}\gamma+\frac{3}{2}\gamma i,\gamma)\) on the lower mirror line and the point \((S+\frac{2Sh_i}{\gamma}+\frac{3}{4}(h_i+\gamma(-1+2i)),\gamma+h_i)\) on the upper mirror line. The lines through these points and the next pinhole $c^{2i-1}$ intersect \(\mu^i\) in the points \((\frac{3i}{2} - S,0)\) and \((\frac{3i}{2}-S-a_i,0)\), respectively. 

	Thus, the set of points on the $i$th enumeration gadget that a valid simplification with $2i+1$ vertices can reach encode all possible subsets of $\{a_{1},\dots,a_{i}\}$.
\end{proof}
From these lemmas it follows that any point that can be reached on the last enumeration gadget using at most \(2m+1\) vertices encodes a sum of the integers in a subset of $A$. We specifically constructed the final enumeration gadget to have its rightmost point (i.e. the endpoint of \(P\)) to be the point that encodes the value $M$. Thus there is a simplification curve with $2m+1$ vertices if and only if there is a subset $B\subset A$ with the total sum of its integers being $M$. We conclude with the following theorem.

\begin{theorem}
	\label{thm:unibend}
	Given a curve distance measure, if there exists a pinhole gadget that can be inserted into the described template such that the above properties hold, the curve-restricted GCS problem for that distance measure is NP-hard.
\end{theorem}

\subsection{NP-hardness of the GCS problem under $\Frechet_{\C}(P,\delta), \protect\diHausdorffP_{\C}(P,\delta) $}
Now, we will prove two versions of the GCS problem to be NP-hard using our template by constructing suitable pinhole gadgets.

\subsubsection{NP-Hardness of Fr\'echet Simplification $\Frechet_{\C}(P,\delta)$}
\label{sec:NPfrechet}
\begin{figure}
	\centering
	\includegraphics{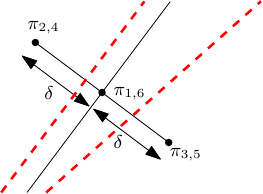}
	\caption{
		A pinhole gadget for the Fr\'echet distance (with even index). Two invalid simplification segments shown as dashed lines. Any simplification that passes below the pinhole will have a distance greater than \(\delta\) to \(\pi_2\). Any simplification that passes above it has distance greater than \(\delta\) to \(\pi_3\).
	}
	\label{fig:pinhole}
\end{figure}%
We construct the pinhole gadget for $\Frechet_{\C}(P,\delta)$ in the following way:
It consists of a chain of five line segments that starts and ends in the pinhole point $c$. The vertices of pinhole gadget \(\pi\) have the following coordinates, \emph{relative to \(c\)} (refer to Figure~\ref{fig:pinhole}): 
%
%
%
\[
\begin{aligned}
\pi_1, \pi_6 & = \left(0, 0 \right)\,,\\
\pi_2, \pi_4 & = \left(\frac{-4\delta}{5}, \frac{3\delta}{5} \right)\,,\\
\pi_3, \pi_5 & = \left(\frac{4\delta}{5}, \frac{-3\delta}{5} \right)\,.
\end{aligned}
\]
These coordinates are for pinhole gadgets with an even index. For pinhole gadgets with an odd index, the sign of the $y$-coordinates is flipped.
\begin{lemma}
	\label{lem:frechetpinhole}
	A segment traversing a pinhole gadget \(\pi\), where \(\pi\) is constructed as described in this subsection, can only be within Fr\'echet distance \(\delta\) of \(\pi\) if it intersects \(\pi\)'s origin (the pinhole).
\end{lemma}
\begin{proof}
	The distance between the outer vertices of the pinhole gadget and the pinhole is exactly \(\delta\). This means there is a segment that passes through the origin that is within Fr\'echet distance $\delta$ of the pinhole gadget $\pi$. Any simplification segment that passes below the origin will have a distance greater than \(\delta\) to \(\pi_3\). Any simplification that passes to the above the origin has distance greater than \(\delta\) to \(\pi_2\).
\end{proof}
To assemble the hardness construction we translate the pinhole gadget to the middle of each of the zig-zag segments of the template curve, and combining Theorem~\ref{thm:unibend} with the lemma above we obtain the following result:

\begin{theorem}
	\label{thm:unipinhole}
	The GCS problem for \(\mathsf{F}_{\C}(P,\delta)\) is NP-hard. 
\end{theorem}

\subsubsection{NP-Hardness of directed Hausdorff Simplification \tinyspace$\protect\diHausdorffP_\C(P,\delta)$}
\begin{figure}
	\centering
	\includegraphics{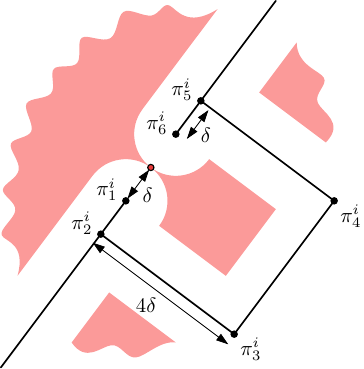}
	\caption{Pinhole gadget for the Directed Hausdorff distance from the simplification to the original curve. The pinhole point is shown in red. The points of the area shaded in light-red have distance greater than \(\delta\) to the curve. Thus no simplification segments can pass through it.}
	\label{fig:dihausgadget}
\end{figure}
Proving NP-hardness for $\protect\diHausdorffP_\C(P,\delta)$ requires another type of pinhole gadget. 
Our new gadget is shown in Figure~\ref{fig:dihausgadget}. Its vertices are as follows, \emph{relative to \(c\)}:

For pinhole gadgets with an even index:
\[
\begin{aligned}
\pi_1 &= \left(\frac{-3\delta}{5},\frac{-4\delta}{5}\right)\,, & \pi_2 &= \left(\frac{-6\delta}{5},\frac{-8\delta}{5}\right)\,, & \pi_3 &= \left(2\delta,-4\delta\right)\,,\\
\pi_4 &= \left(\frac{22\delta}{5},-\frac{-4\delta}{5}\right)\,, & \pi_5 &= \left(\frac{6\delta}{5},\frac{8\delta}{5}\right)\,, & \pi_6 &= \left(\frac{3\delta}{5},\frac{4\delta}{5}\right)\,.
\end{aligned}
\]
For the odd-indexed pinhole gadgets this construction must be mirrored around a horizontal line through the pinhole, so the line segment \(P[\pi_1,\pi_2]\) runs parallel with the line segment from the preceding split gadget to \(\pi_1\).
\begin{lemma}
	\label{lem:dihausgadget}
	A segment traversing a pinhole gadget \(\pi\), where \(\pi\) is constructed as described in this subsection,  can only be within directed Hausdorff distance (directed from the segment to \(\pi\)) \(\delta\) if it intersects \(\pi\)'s origin (the pinhole).  
\end{lemma}
\begin{proof}
	\(\pi^i_1\) and \(\pi^i_6\) both have distance exactly \(\delta\) to the origin so a segment that passes through it has distance \(\delta\) to the gadget. Trying to pass to the right of the origin doesn't work as there is a gap of size \(4\delta\) between the gadget's segments on opposite sides of the origin. On the left side of the origin there are no segments at all that could have distance less than \(\delta\).
\end{proof}
To assemble the hardness construction we translate the pinhole gadget to the middle of each of the zig-zag segments of the template curve (rotating the gadget as needed as stated above). Combining Theorem~\ref{thm:unibend} with the lemma above we obtain the following result:

\begin{theorem}
	\label{thm:dihaushard}
	The GCS problem for \(\overrightarrow{\mathsf{H}}_{\C}(P,\delta)\) is NP-hard. 
\end{theorem}

\subsection{Extending the template}
With our template construction, it is also possible to prove NP-hardness for curve restricted GCS problem under the weak Fr\'echet, directed Hausdorff (directed from curve to simplification), and (undirected) Hausdorff distance measures. This requires new types of pinhole gadgets, but these can be difficult to design. We will therefore use an alternative approach and show how we can combine our template with the gadget introduced for Fr\'echet distance in Section \ref{sec:NPfrechet} to prove NP-hardness for these distance measures, by expanding one of the properties of our construction. To see why the template and this gadget do not trivially imply NP-hardness for these measures by the steps shown earlier in this section, refer to Figure \ref{fig:pinholefail}: Under distance measures like weak Fr\'echet, segments of \(P'\) do not have to pass exactly through the pinhole but can also pass within a small (Euclidean) distance of the pinhole.
This implies that from a single vertex on a split gadget, instead of just one point being reachable on the next enumeration gadget, there is now an interval of reachable points. For each of the points in this interval there is an interval of reachable points on the next split gadget. The union of these intervals gives a bigger interval and so the size of the intervals is increasing as they are propagated throughout the construction.

Each interval on an enumeration segment still contains the point that precisely encodes a subset sum, along with other points whose $x$-coordinates are within some small amount of this subset sum. If the intervals encoding the sums of different subsets stay small enough so that they never overlap with each other, the NP-hardness construction still holds. In this case the endpoint of \(P\) is contained in the interval corresponding to a subset of $A$ with the sum of elements equal to exactly \(M\).
Therefore we can replace the first property required for our construction (Any segment of \(P'\) starting before a pinhole gadget and ending after the pinhole gadget must pass through the pinhole gadget's \emph{pinhole}) with a more easily satisfied property below and still have proof of NP-hardness:
\begin{enumerate}
	\item The endpoint of any segment of \(P'\) starting before a pinhole gadget and ending after the pinhole gadget must have distance less than \(\frac{0.5}{2m}\) to the endpoint of the segment with the same starting point that passes exactly through the pinhole and ends on the same segment of \(P\).
\end{enumerate} 

Since an interval can be propagated at most \(2m\) times (once for each pinhole), the width of the interval on the final enumeration gadget is always less than \(0.5\) on either side of the point reached by only going through pinholes. This means 
the intervals will not overlap, since the points reached by only going through pinholes have distance 1 to each other. This means the construction still implies NP-hardness.

We will now show that this property holds for weak Fr\'echet distance, directed Hausdorff distance from curve to simplification, and undirected Hausdorff distance for the pinhole gadget depicted in Figure~\ref{fig:pinhole}.
\begin{lemma}
	\label{lem:exppin}
	The expanded pinhole property holds for weak Fr\'echet distance, directed Hausdorff distance from curve to simplification, and undirected Hausdorff distance.
\end{lemma} 
\begin{proof}
	For these distance measures, simplification segments from an enumeration gadget to a split gadget (and vice versa) are valid iff the segment has distance less than \(\delta\) to all of the vertices of the pinhole gadget in between. As shown in the proof for Lemma~\ref{lem:encodesubsetsum}, a segment that starts on enumeration gadget \(\mu^{i-1}\) at the point \(v_{\mu^{i-1}}=(\frac{3(i-1)}{2}\gamma - S,0)\) that goes through the pinhole will intersect the split gadget \(\sigma^i\) in the point \((S + \frac{3i}{2}\gamma - \frac{3}{4}\gamma,\gamma)\). So, if our intervals are to remain small enough it must be impossible to have a segment from \(v_{\mu^{i-1}}\) to the point \(v_{\sigma^i} = (S + \frac{3i}{2}\gamma - \frac{3}{4}\gamma + \frac{0.5}{2m},\gamma)\) with distance \(\leq \delta\) to the pinhole gadget. The distance from the segment between these points to the point \(\pi^{2(i-1)}_2\) (the leftmost vertex of the pinhole gadget) can be given by the equation
	\[
	D(P\langle v_{\mu^{i-1}},v_{\sigma^i}\rangle,\pi^{2(i-1)}_2)=\frac{|\frac{5 \gamma + \delta (6 + 50 \gamma m + 48 m S)}{m}|}{10 \sqrt{\frac{25 \gamma^2 m^2 + 6 \gamma m (1 + 8 m S) + (1 + 8 m S)^2}{m^2}}}\,,
	\]
	where \(D(\cdot,\cdot)\) is Euclidean distance.
	Given the values we chose for \(\delta\) and \(\gamma\) in our construction (\(\delta = 4\sum_{a\in A}a, \gamma = \delta^2\)), this distance is greater than \(\delta\). The distance between the line segment starting on \(v_{\mu^{i-1}}\) and ending in \((S + \frac{3i}{2}\gamma - \frac{3}{4}\gamma - \frac{0.5}{2m},\gamma)\) to \(\pi^{2(i-1)}_3\) is also greater than \(\delta\). This same argument also applies to line segments ending in the upper mirror segment, and to segments from a split gadget to an enumeration gadget.
\end{proof}

Combining Lemma~\ref{lem:exppin} and Theorem~\ref{thm:unibend} gives the following result:
\begin{theorem}\label{thm:CR3D}
	The GCS problem for \(\mathsf{wF}_{\C}(P,\delta), \overleftarrow{\mathsf{H}}_{\C}(P,\delta), \mathsf{H}_{\C}(P,\delta)\) is NP-hard.
\end{theorem}

In Section~\ref {sec:NRHD} (Corollary~\ref{cor:CRHD}) we also prove strong NP-hardness for curve-restricted GCS under undirected Hausdorff distance. 

\begin{figure}
	\centering
	\includegraphics{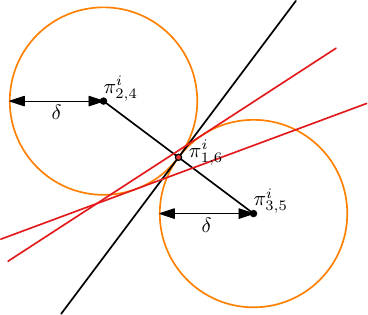}
	\caption{The pinhole gadget. The orange circles indicate all points with distance \(\leq \delta\) to the vertices $\pi^{i}_{1}=\pi^{i}_{3}$ and $\pi^{i}_{2}=\pi^{i}_{4}$ respectively. Since weak Fr\'echet distance allows for non-monotonic mapping, this implies that the segments of valid simplification are not forced to pass through the pinhole. As long as a simplification segment intersects both circles, it covers the pinhole gadget.
	}
	\label{fig:pinholefail}
\end{figure}

\section{Curve-restricted GCS under Discrete Fr\'echet Distance} \label{subsec:CRDF}
In this section we present an $O(n^3)$ time algorithm for $\dFrechet_{\C}(P,\delta)$. 
Let $P = \langle p_1, \ldots, p_n \rangle$ be a polygonal curve.  We
first argue that there is only a discrete set of candidate points we
need to consider for vertices of the output curve.  Let $\cal A$ be
the arrangement of $n$ disks of radius $\delta$ centered on the points
in $P$, and let $C = \langle c_1, \ldots, c_m \rangle$ be the sequence
of intersections between the polyline $P$ and $\cal A$, in order of
$P$.

\begin {observation}
Under the discrete \Frd,
if there exists a curve-restricted simplification $P' = \langle q_1, \ldots, q_k \rangle$ of $P$, then there exists a subsequence of $C$ of length $k$ which is a simplification of $P$.
\end {observation}

Clearly, $C$ consists of at most $m \in O(n^2)$ points.
Bereg \etal~\cite{bjwyz-s3pcudfd-08} show how to compute the minimal vertex-restricted simplification of $\cal A$. We cannot apply their result directly by treating all points in $C$ as vertices, since we do not require the simplification to be mapped to all such points, only those in $P$. However, we can design an algorithm in similar fashion.

Define $K(i,j)$ to be the minimum value $k$ such that there exists a subsequence $c_1, \ldots, c_j$ of length $k$ that has discrete Fr\'echet distance at most $\delta$ to the sequence $p_1, \ldots, p_i$. If no such sequence exists at all, we set $K(i,j)=\infty$ (note that this happens if and only if the distance between $p_i$ and $c_j$ is larger than $\delta$).
We will design a dynamic program to calculate all $nm$ values $K(i,j)$.

First, we observe that we only need to match every point in $P$ once. In principle, the discrete Fr\'echet distance is defined by a sequence of pairs of points, one from each sequence, subsequent pairs can advance in either of the two sequences. However, we never need to consider the second case here.  
\begin {observation}
For the optimal solution $P'$ of length $k$, there exists a Fr\'echet matching consisting of exactly $n$ pairs of points, each matching a unique point in $P$ to a point in $P'$.
\end {observation}
\begin {proof}
Suppose there would be two consecutive pairs $(p_i,q_j)$ and $(p_i,q_{j+1})$ matching the same point $p_i$ in $P$ to different consecutive points $q_j$ and $q_{j+1}$ in $P'$. Since $P'$ is a minimum-length subsequence of $C$, both $q_j$ and $q_{i+j}$ must also match to another points: there must be a pair $(p_{i-1},q_j)$ (otherwise we could eliminate the point $q_j$), and there must be a pair $(p_{i+1},q_{j+1})$ (otherwise we could eliminate the point $q_{j+1}$.
But then, one of the pairs $(p_i,q_j)$ and $(p_i,q_{j+1})$ is superfluous: we can remove either one and still have a valid Fr\'echet matching between $P$ and $P'$.
\end {proof}
This observation implies that when calculating the value of $K(i,j)$, we only have to consider values of the form $K(i-1,j')$, where $1 \le j' \le j$.
Specifically, if $p_{i-1}$ and $c_j$ are within distance $\delta$, then 
\[
K(i,j) = \min \left( K(i-1,j), \min_{1 \le j' < j} \left( K(i-1,j') + 1\right) \right)
\quad \mathrm {\ if\ distance\ } < \delta, \mathrm {\ otherwise\ } \infty.
\]
This definition immediately gives an $O(n^4)$ algorithm to compute $K(n,m)$.
We can improve on this by maintaining a second table with prefix minima. Let $M(i,j) = \min_{1 \le j' \le j} K(i,j)$. Then we have the recursive system
\[
K(i,j) = \min \left( K(i-1,j), M(i-1,j-1) + 1 \right) \;\mbox{and}\;\;
M(i,j) = \min \left( M(i,j-1), K(i,j) \right),
\]
%
which can clearly be calculated in constant time per table entry, and overall saves a linear factor.

\begin{theorem} \label{thm:CRDF}
	Given a polygonal curve $P$ with $n$ vertices and $\delta>0$, an optimal solution to the $\dFrechet_{\C}(P,\delta)$ can be computed in $O(n^3)$ time and $O(n^2)$ space.
\end{theorem}
It is an interesting question whether this time bound can be improved. A more greedy approach fails because of the requirement that $P'$ is a subsequence of $C$, and not just a subset: a local choice to advance might have implications later. 


\section{A Linear Time Algorithm for Curve-restricted Fr\'echet Distance in $\Reals^1$}  \label{subsec:CRFD1}
In this section we provide a greedy algorithm for the curve-restricted GCS problem in $\Reals^1$ under the Fr\'echet distance.
We describe our algorithm using the man-dog terminology that is often used in the literature on Fr\'echet distance:
Initially a man and his dog start at $p_1$. The man walks along $P$ until his distance to the dog exceeds $\delta$. Now if there is a turn between the man and the dog, the dog marks its current position and jumps over the turn and stays at distance exactly $\delta$ away from the man. If there is no turn in between, the dog just follows the man at distance exactly $\delta$ and stops when the man arrives at the next turn or the end. Once they both end the walk at $p_n$ we report the positions marked by the dog as $P'$. See \figref{fig:Frechet1D}. We have the following theorem:

\begin{figure} [t!]
	\centering
	\includegraphics{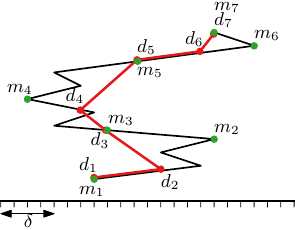}
	\caption{The time-stamped traversal made by the man $m$ and the dog $d$. The red lines indicate the jumps by the dog.}
	\label{fig:Frechet1D}
\end{figure}
\begin{theorem} \label{thm:CRFD1}
	Given a polygonal curve $P$ in $\Reals^{1}$ with $n$ vertices and $\delta>0$, an optimal solution to the curve-restricted GCS problem under Fr\'echet distance can be computed in linear time.
\end{theorem}

\begin{proof}
	Let $Q=\langle q_1,\cdots,q_k \rangle$ be the optimal simplified curve and let $P'=\langle p'_1,\cdots,p'_m\rangle$ be the curve returned by the algorithm. For the sake of a contradiction, we assume that $k<m$. Observe that every time $\mathsf{Dir}$ changes (some turn occurs) and $\|p_i-s\|>\delta$, (at least one and at most two) vertex in $P'$ will be added. Both curves have $p_1$ and $p_n$ as their first and last vertices, respectively, in common. The remaining argument is that $Q$ has fewer internal vertices than $P'$ made by the two conditional loops above. The algorithm incrementally updates $P'$ in a way that when $\textsf{Dir}$ changes always a vertex will be added to $P'$ that has the distance exactly $\delta$ to the turn vertex in $P$. Observe that if $Q$ tends to have fewer number of vertices should skip one of this critical vertices which results in having distance greater than $\delta$ to turn vertices in $P$. Therefore $\Frechet(P,Q)>\delta$ and contradiction. 
\end{proof}

Here we present the pseudocode of the algorithm that we verbally proposed. $x(.)$ indicates the $x$-coordinate of points.

\begin{algorithm} [htpb]
	\DontPrintSemicolon
	\SetKwFor{ForAll}{for each}{:}{}
	\caption{Algorithm for curve-restricted GCS problem for $\Frechet_\C(P,\delta)$ in $\Reals^1$.}
	\label{alg:CRFD1}
	
	$i,j \leftarrow 1$,~
	$s \leftarrow p_i$,~
	$p'_j \leftarrow s$,~
	$j\leftarrow j+1$,~
	$i \leftarrow i+1$\;
	\lIf{$x(p_i)>x(s)$}
	{$\mathsf{Dir}\leftarrow \mbox{`right'}$}
	\lElse {$\mathsf{Dir}\leftarrow \mbox{`left'}$}
	\For{$i=2$ \KwTo {$n$}}
	{
		\lIf{$i\geq n$}{$p'_j\leftarrow p_n$}
		\Else{
			\If{$x(p_i)\geq x(s)$}
			{	
				\lIf{$\|p_i-s\|\leq \delta$}{{\bf Continue}}
				\Else{
					\lIf{$\mathsf{Dir}=\mbox{`left'}$}
					{$p'_j\leftarrow s$,~$s\leftarrow p_i-\delta$,~$j\leftarrow j+1$
					}
					\lElse{
						$s\leftarrow p_i-\delta$
					}
				}
				$\mathsf{Dir}\leftarrow \mbox{`right'}$
			}
			\Else
			{
				
				\lIf{$\|p_i-s\|\leq \delta$}{{\bf Continue}}
				\Else{
					\lIf{$\mathsf{Dir}=\mbox{`right'}$}{
						$p'_j\leftarrow s$,~$s\leftarrow p_i+\delta$,~$j\leftarrow j+1$ 
					}
					\lElse{
						$s\leftarrow p_i+\delta$
					}
				}
				$\mathsf{Dir}\leftarrow \mbox{`left'}$\;
			}	
			
		}
	}
	\KwRet{$\langle p'_1 ,\cdots, p'_{j}\rangle$}

\end{algorithm}


\section{Approximation Algorithm for Non-restricted GCS under the Fr\'echet Distance} \label{sec:NRFD}
In this section we present a simple approximation algorithm for the
non-restricted variant of GCS problem that discretizes the feasible space
for the vertices of the simplified curve. 
The idea is to compute a polynomial number of shortcuts in the discretized space, and (approximately) validate for each shortcut whether it is within \Frd\ $\delta$ to a subcurve of $P$. For every subcurve of $P$ we incrementally add the valid shortcuts to the edge set of a graph $G$ until all the shortcuts have been processed. Once $G$ is built, we compute the shortest path in $G$ and return $P'$.  To speed up the validation for each shortcut, we use a data structure to decide whether the \Frd\ between a shortcut and a subcurve of $P$ is at most $\delta$. For a better understanding of our algorithm, we introduce some notation. Consider a ball $B(o,r)$ of radius $r > 0$ centered at $o\in \Reals^d$.
%
Let ${\partition(\Reals^d, l)}$ be a {\em partitioning} of $\Reals^d$ into a set of disjoint cells (hypercubes) of side length $l$ that is induced by axis parallel hyperplanes placed consecutively at distance $l$.
For any $1\leq i\leq n$ we call ${\C_i = \C_i(r, l)} = \{c\in \partition(\Reals^d, l)\,|\, c\cap B(p_i,r)\neq\emptyset\}$ a \emph{discretization} of $B(p_i,r)$. Let $\G_i$ be the set of {\em corners} of all cells in $\C_i$.


\begin{figure}[t]
	\centering
	\includegraphics[width=0.16\textheight]{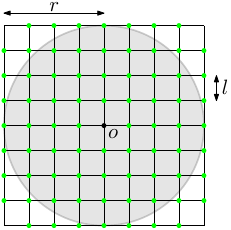}
	\caption{$G_o$ as the collection of cells in $\partition(\Reals^d,l)$ that intersect $B(o,r)$, and the set of corners highlighted in green.}
	\label{fig:discretization}
\end{figure}

\begin{algorithm}[htbp]
	\LinesNotNumbered
	\SetKwFunction{MinLinkSimp}{\textsc{MinLinkSimp{$(P,\delta,\eps$)}}}
	\SetKwFunction{validate}{\textsc{Validate}}
	\caption{Algorithm for non-restricted GCS problem for Fr\'echet distance.}
	\label{alg:NRFD}
	
	\BlankLine 
	
	\lForAll{$i \in \{1,\cdots, n\}$}
	{
		Compute $\C_i(\delta,\big(\eps\delta/(4\sqrt{d})\big)$ and $\G_i$\label{step:partition} 
	} 
	$E\leftarrow \emptyset,$~
	$V \leftarrow \emptyset,$~
	$\C_1\leftarrow p_1\cup \C_1,$~
	$\C_n\leftarrow p_n\cup \C_n$\;
	\ForAll{$\C_i~\mbox{{\bf and}}~\C_j$, {\bf with} $1\leq i\leq j \leq n$} 
	{\ForAll{$c_1 \in \C_i~\mbox{{\bf and}}~c_2 \in \C_j$} 
		{ \lIf{$\validate({\segment{c_1c_2}},P[i,j]) = true$}
			{$E \leftarrow E \cup {\segment{c_1c_2}},~V\leftarrow V\cup \{c_1,c_2\}$} \label{step:edgeSetij}
		}
	}
	
	\KwRet the shortest path between $p_1$ and $p_n$ in $G=(V,E)$.
	
\end{algorithm}

As we can see \algref{alg:NRFD} is a straightforward computation of valid shortcuts and shortest path in the shortcut graph $G$. The \validate procedure takes a shortcut $\segment{c_1c_2}$ and a subcurve $P[i,j]$ as arguments and its task is to (approximately) decide whether $\Frechet(\segment{c_1c_2},P[i,j])\leq \delta$ or not. In particular, it returns {\em true} if $\Frechet(\segment{c_1c_2},P[i,j])\leq (1+\eps/2)\delta$ and {\em false} if $\Frechet(\segment{c_1c_2},P[i,j]) > (1+\eps)\delta$. We efficiently implement the \validate procedure (lines~\ref{step:edgeSetij}) by means of the data structure in~\cite{dh-jydcf-13}. 
Here we slightly rephrase the theorem that refers to the data structure in~\cite{dh-jydcf-13} according to our terminology:

\begin{lemma}
	[Theorem 5.9 in~\cite{dh-jydcf-13}] \label{lem:dhStructure}
	Let $P$ be a polygonal curve in $\Reals^d$ with $n$ vertices and $0<\eps\leq 1/8$ a real value. One can construct a data structure of size $O\big((\eps^{-d}\log^2(1/\eps))n\big)$ and construction time of $O\big((\eps^{-d}\log^2(1/\eps))n\log^2 n\big)$, such that for any query segment $\segment{ab}$ in $\Reals^d$ and two vertices $p_i$ and $p_j$ in $P$ with $1\leq i \leq j \leq n$, one can compute a $(1+\eps)$-approximation of $\Frechet(\segment{ab},P[i,j])$ in $O(\eps^{-2}\log n \log \log n)$ query time.
\end{lemma}

\begin{lemma}
	\label{lem:validate}
	Let $0<\eps\leq 1$ and let $\segment{ab}$ be a segment in $\Reals^d$ such that $a$ and $b$ are confined within some two cells $h' \in \G_i$ and $h''\in \G_j$, respectively, with $1 \leq i \leq j \leq n$. 
	\begin{enumerate}
		\item If $\Frechet(\segment{ab},P[i,j])\leq \delta$, then for all corners $c' \in h'$ and $c'' \in h''$ $\validate(\segment{c'c''},P[i,j])$ returns `true'. \label{prop:YesDecision}
		
		\item If $\validate(\segment{c'c''},P[i,j])$ returns `false' for all corners $c' \in h'$ and $c'' \in h''$, then   $\Frechet(\segment{ab},P[i,j]) > (1+\eps/4) \delta$. \label{prop:NoDecision}
		
	\end{enumerate} 
\end{lemma}

\begin{proof} 
	\ref{prop:YesDecision}. Let $c'$ be an arbitrary corner of $h'$ and $c''$ an arbitrary corner
	of $h''$. Note that $\mathsf{Diam}(h')=\mathsf{Diam}(h'')=
	\sqrt{d}\cdot(\eps\delta/4\sqrt{d})= \eps\delta/4$, where
	$\mathsf{Diam}(h')$ and $\mathsf{Diam}(h'')$ are the diameter of cells
	$h'$ and $h''$, respectively. Hence $\ell_1= \|a-c'\|\leq (\eps/4)
	\delta$ and $\ell_2= \|b-c''\|\leq (\eps/4)\delta$. Given the two
	segments $\segment{ab}$ and $\segment{c'c''}$ the \Frd\ between them
	is $\Frechet(\segment{c'c''},\segment{ab})=\max
	\{\ell_1,\ell_2\}=(\eps/4)\delta$ by~\cite{ag-cfdb-95}. Now by
	applying a triangle inequality between the segments and path $P[i,j]$
	we have: \[\Frechet(\segment{c'c''},P[i,j])\leq
	\Frechet(\segment{ab},P[i,j])+\Frechet(\segment{c'c''},\segment{ab})
	\leq \delta+\eps\delta/4 = (1+\eps/4)\delta.\] We build the
	data structure of \lemref{lem:dhStructure} for $\validate$ with
	respect to parameter $\eps/6$ and the whole curve $P$ that leads
	to a $(1+\eps/6)$-approximation of the \Frd\ between
	$\segment{c'c''}$ and $P[i,j]$ returned by
	$\validate(\segment{c'c''},P[i,j])$. We denote this approximation by
	$D$ and have:
	\[D \leq (1+\eps/6)\cdot\Frechet(\segment{c'c''},P[i,j]) \leq (1+\eps/6)(1+\eps/4)\delta= (1+10\eps/24 + \eps^2/24)\delta \leq  (1+\eps/2)\delta,\] for $\eps\leq 1$.
	Therefore, for all corners $c' \in h'$ and $c'' \in h''$
	$\validate(\segment{c'c''},P[i,j])$ returns `true' as claimed.
	
	\ref{prop:NoDecision}. If $\validate(\segment{c'c''},P[i,j])$ returns 'false' then we know that the value $D$ returned by the data structure is at least: 
	
	\[D \geq \Frechet(\segment{c'c''},P[i,j]) >  (1+\eps/2)\delta.\] Now by
	applying a triangle inequality between the segments and path $P[i,j]$
	we have: \[(1+\eps/2)\delta - \eps\delta/4 <\Frechet(\segment{c'c''},P[i,j]) - \Frechet(\segment{c'c''},\segment{ab}) \leq
	\Frechet(\segment{ab},P[i,j]), \] therefore $\Frechet(\segment{ab},P[i,j])> (1+\eps/4)\delta$. This completes the proof.
\end{proof}

\begin{lemma} \label{lem:twoOPT}
	There exists a $P'=\Frechet_{\N}(P,\delta)$ such that every link $\segment{ab}$ in $P'$ does not match to a proper subsegment of $P[i,i+1]$, i.e., $P(i,i+1)$ for some all $1< i<n$.
\end{lemma}

\begin{proof} 
	We use proof by contradiction. Assume that such a $P'$ fulfilling the
	condition above does not exist, hence there is a link $\segment{ab}$ whose endpoints
	match to neither $p_i$ nor $p_{i+1}$, for some $1< i<n$.
	Let $(\sigma,\theta)$ be a Fr\'echet matching realizing
	$\Frechet(P,P')\leq \delta$ and let $t_i$ and $t_{i+1}$ be two real
	values with $0\leq t_i<t_{i+1}\leq 1$ such that $\sigma(t_i))=i$,
	$\sigma(t_{i+1})={i+1}$ and $\theta(t_i)=x$, $\theta(t_{i+1})=y$, for
	some $1<x < y<\#P'+1$. Note that since $\segment{ab}$ matches to neither
	$p_i$ nor $p_{i+1}$, thus $x< a< b < y$. Now let $P''=P'[1,x]\circ
	\segment{xy} \circ P'[y,n]$. Clearly, $\Frechet(P,P'')\leq \delta$ due
	to $\Frechet(P[i,i+1],\segment{xy})\leq \delta$ and also $P''$ has the
	same number of vertices as $P'$ has (see
	\figref{fig:twoOPT}). Therefore, $P''=\Frechet_{\N}(P,\delta)$
	exists and this is a contradiction.
	\begin{figure} [h]
		\centering
		\includegraphics[width=.45\linewidth]{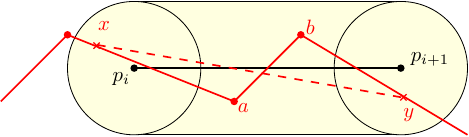}
		\caption{The shortcut $\segment{xy}$ has the Fr\'echet distance at most $\delta$ to $P[i,i+1]$.}
		\label{fig:twoOPT}
	\end{figure}%
\end{proof}
\begin{lemma} \label{lem:sublinkExists}
	Let $\segment{ab}$ be an arbitrary link in $P'=\Frechet_{\N}(P,\delta)$ fulfilling the condition in~\lemref{lem:twoOPT}. The following statements are true: 
	\begin{enumerate}
		\item There exist a sublink $\segment{a'b'} \subseteq \segment{ab}$ and an integer pair $(i,j)$ with $1\leq i \leq j \leq n$ such that $\Frechet(\segment{a'b'},P[i,j])\leq \delta$. \label{prop:sublink1} 
		\item There exist an integer pair of cells $(h',h'')$ and an integer pair $(i,j)$ with $1\leq i \leq j \leq n$ and $h' \in \G_i$, $h'' \in \G_j$, such that $a'$ and $b'$ are confined within $h'$ and $h''$, respectively and $\validate(\segment{c'c''},P[i,j])$ returns `true' for all corners $c' \in h'$ and $c'' \in h''$. \label{prop:sublink2}
	\end{enumerate} 
	
\end{lemma}

\begin{proof} 
	(\ref{prop:sublink1}) 
	Let $P'$ be a solution to $\Frechet_{\N}(P,\delta)$ satisfying~\lemref{lem:twoOPT}.
	Then $\segment{ab}$ as an arbitrary link in $P'$ cannot match to a proper subsegment in $P[i,i+1]$ for all $1< i<n$. Now let $(\sigma,\theta)$ be a Fr\'echet matching realizing $\Frechet(P,P')\leq \delta$ and let $t_a$ and $t_b$ be two real values with $0\leq t_a<t_b\leq 1$ such that $\sigma(t_a)=p$, $\sigma(t_b)=q$ and $\theta(t_a)=a$, $\theta(t_b)=b$. We then have two cases: 
	\begin{itemize}
		\item $\segment{ab}$ only intersects $B(p_i,\delta)$. Now let $0\leq t_i<1$ be a value such that $\sigma(t_i)=i$. Then clearly, $\theta(t_a)\leq \theta(t_i) \leq \theta(t_b)$ and correspondingly, $\sigma(t_a)\leq \sigma(t_i)= i \leq \sigma(t_b)$. Hence, it holds that $1 \leq p \leq i \leq q \leq n$ for some $1< i <n$. Therefore there exists a sublink $\segment{a'b'}= B(p_i,\delta)\cap \segment{ab}$ and an integer pair $(i,j)$ with $i=j$ such that $\Frechet(\segment{a'b'},P[i,j])\leq \delta$ (see \figref{fig:sublinkExists} (a)).

		\item $\segment{ab}$ intersects all the balls $B(p_i,\delta), \cdots,B(p_j,\delta)$ for some $1\leq i<j\leq n$ in order from $i$ to $j$. Now let $0\leq t_i<t_j\leq 1$ be a value such that $\sigma(t_i)=i$ and $\sigma(t_j)=j$. Then clearly, $\theta(t_a)\leq \theta(t_i)  < \theta(t_j) \leq \theta(t_b)$ and correspondingly, $\sigma(t_a)\leq \sigma(t_i)  < \sigma(t_j) \leq \sigma(t_b)$. Hence, it holds that $1 \leq p \leq i < j \leq q \leq n$ for some $1\leq i < j \leq n$. Now let $\segment{a'b'}= P'[\theta(t_i),\theta(t_j)]$. Following the Fr\'echet matching $(\sigma,\theta)$, we have $\Frechet(\segment{a'b'}, P[i,j])\leq \delta$ (see \figref{fig:sublinkExists} (b)).
	\end{itemize}
	Therefore such a $\segment{a'b'}\subseteq \segment{ab}$ and pair $(i,j)$ with $1\leq i\leq j \leq n$ exists.
	
	(\ref{prop:sublink2}) By Property~\ref{prop:sublink1}, we know that $\segment{a'b'}$ and an integer pair $(p_i,p_j)$ with $1\leq i \leq j \leq n$ exist such that $\Frechet(\segment{a'b'},P[i,j])\leq \delta$. Plugging this into~\lemref{lem:validate} completes the proof.	
	\begin{figure} [htbp]
		\centering
		\includegraphics[width=.6\textheight]{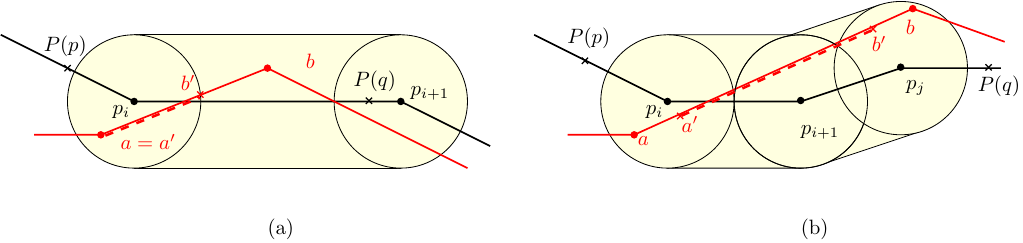}
		\caption{(a) Any point on the subsegment $\segment{a'b'}$ can be matched to $p_i$, thus $\Frechet(\segment{a'b'},P[i,j]\leq \delta$ with $i=j$. (b) Following the optimal matching, $\Frechet(\segment{a'b'},P[i,j])\leq \delta$ with $i < j$.}
		\label{fig:sublinkExists}
	\end{figure}
\end{proof}

The following lemmas allow us to conclude~\thmref{thm:NRFD}.

\begin{lemma} \label{lem:SPExists}
	The shortest path $P'_\mathsf{alg}$ returned by \algref{alg:NRFD} exists and $\Frechet(P,P'_\mathsf{alg})\leq (1+\eps)\delta$. 
\end{lemma}

\begin{proof} 
	Let $P'=\Frechet_{\N}(P,\delta)$, $\segment{ab}$ and $\segment{bc}$ be two consecutive links in $P'$. By \lemref{lem:sublinkExists} (Property~\ref{prop:sublink2}) there exist an integer pair $(i,j)$ with $1\leq i\leq j\leq n $, and two corners $c_i\in \G_i$ and $c_j\in G_j$ close to a sublink $\segment{xw}$ of $\segment{ab}$ such that $\validate(\segment{c_ic_j},P[i,j])$ returns `true'. We have two followng cases:
	
	\begin{itemize}
		\item if $b \in B(p_j,\delta)$: then there exist an integer $k$ with $1\leq j\leq k\leq n $, and two corners $c'_{j}$ and $c_k$ close to a sublink $\segment{yz}$ of $\segment{bc}$ such that $\validate(\segment{c'_{j}c_k},P[j,k])$ returns `true'. Since $P'[x,y]$ has only one intermediate vertex i.e., $b$ that lies within $B(p_j,\delta)$, it is easy to see that $\Frechet(\segment{xy},p_j)\leq \delta$ and therefore, $\validate(\segment{c_{j}c'_{j}},p_j)$ returns `true' as well.
		
		\item  if $b \notin B(p_j,\delta)$: then there exits an interger $k$ with $1\leq j+1 \leq k \leq n$, and two corners $c_{j+1}$ and $c_k$ close to a sublink $\segment{yz}$ of $\segment{bc}$ such that $\validate(\segment{c_{j+1}c_k},P[j+1,k])$ returns `true'. Now again  $P'[x,y]$ has only one intermediate vertex i.e., $b$ that lies within $C(\segment{p_jp_{j+1}},\delta)$, thus $\Frechet(\segment{xy},P[j,j+1])\leq \delta$ and therefore, $\validate(\segment{c_{j}c_{j+1}},P[j,j+1])$ returns `true'.
	\end{itemize}   
	This indicates that for any two adjacent links in $P'$ we have three links in the shortcut graph. Since $P'$ is connected the graph has a connected path $P'_{\mathsf{alg}}$. Note that if such a path, $P'_{\mathsf{alg}}$, is found by the algorithm is not the shortest path in $G$, then there must be another path which leads to existing a shortest path either way. The proof for a single link in $P'$ is trivial. To prove $\Frechet(P,P'_\mathsf{alg})\leq (1+\eps)\delta$, let $e$ be a link in $P'_\mathsf{alg}$ and $P_e$ be the corresponding subcurve in $P$ such that $\validate(e,P_e)$ has returned `true'. This implies that $\Frechet(e,P_e)\leq (1+\eps)\delta$, therefore $\Frechet(P,P'_\mathsf{alg})= \max_{e \in P'_\mathsf{alg}}\Frechet(e,P_e) \leq (1+\eps)\delta$.	
\end{proof}

\begin{lemma} \label{lem:numberOfLinks} 
	Let $P'=\Frechet_\N(P,\delta)$ and let $P'_\mathsf{alg}$ be the curve returned by \algref{alg:NRFD}. Then $\#P'_\mathsf{alg}\leq 2(\#P'-1).$
\end{lemma}

\begin{proof} 
	Let $m=\#P'$ where $P'=\Frechet_\N(P,\delta)$. 
	Now for all $m-2$ ($m>2$) intermediate links in $P'$ i.e., $P[2,m]$ the number of links in $P'_\mathsf{alg}[c_2,c_m]$ where $c_1\in B(p_2,\delta)$ and $c_m \in B(q_m,\delta)$ is at most $2(m-2)+1$ in the worst case. This together with the two first and last links $\segment{c_1c_2}$ and $\segment{c_mc_{m+1}}$ in $P'_\mathsf{alg}$, respectively, results in $\#P'_\mathsf{alg}\leq 2m-1$. 
\end{proof}

\begin{lemma} \label{lem:NRFDRuntime}
	\algref{alg:NRFD} runs in $O\Big( \eps^{-d}n\log n\big(\log^2(1/\eps)\log n+ \eps^{-(d+2)}n\log\log n\big) \Big)$ time and uses $O\big((\eps^{-d}\log^2(1/\eps))n\big)$ space. 
\end{lemma}

\begin{proof} 
	The number of cells in each ball is bounded by $O\big((\delta/(\eps\delta/4\sqrt{d}))^d\big)= O(\eps^{-d})$. There are $O(n^2)$ pairs of balls, hence, we have $O(\eps^{-2d}n^2)$ pairs of corners $c'$ and $c''$ to pass to \validate procedure to determine whether $\Frechet(\segment{c'c''},P[i,j])\leq (1+\eps)\delta$ or not for all $c'\in \G_i$ and $c'' \in \G_j$ with $1\leq i\leq j\leq n$. On the other hand by \lemref{lem:dhStructure} we speed up \validate procedure. The construction time takes $O\big((\eps^{-d}\log^2(1/\eps))n\log^2 n\big)$ and its query takes $O(\eps^{-2}\log \log\log n)$ per $\segment{c'c''}$. Having $O(\eps^{-2d}n^2)$ calls on \validate, the whole validation process takes $O((\eps^{-2d+2}n^2 \log n\log \log n))$. Therefore the total runtime would be the sum of the construction time of the data structure of \lemref{lem:dhStructure} and the latter one which is:  $O\Big( \eps^{-d}n\log n\big(\log^2(1/\eps)\log n+ \eps^{-(d+2)}n\log\log n\big) \Big)$ as claimed. The space follows from~\lemref{lem:dhStructure}.
\end{proof}

\begin{theorem} \label{thm:NRFD}
	Let $P$ be a polygonal curve with $n$ vertices in $\Reals^d$, $\delta>0$, and $P'=\Frechet_{\N}(P,\delta)$.
	For any $0<\eps\leq 1$, one can compute in $O^*(n^2\log n \log\log n)$ time and $O^*(n)$ space a non-restricted simplification $P^*$ of $P$ such that $\#P^*\leq 2(\#P'-1)$ and $\Frechet(P,P^*)\leq (1+\eps)\delta$. Here, $O^*$ hides factors polynomial in $1/\eps$.
\end{theorem}

\begin{corollary} \label{cor:NRWFD}
	\thmref{thm:NRFD} also holds for $\wFrechet_{\N}(P,\delta)$.
\end{corollary}
\begin{proof}
	One can modify the data structure in~\lemref{lem:dhStructure} to support queries under the weak Fr\'echet distance by regarding the weak Fr\'echet distance instead of the strong Fr\'echet distance in preprocessing stage (see \cite{dh-jydcf-13} for more details on their data structure). Also following the fact that the triangle inequality holds for the weak Fr\'echet distance, \lemref{lem:validate}, \lemref{lem:twoOPT} and \lemref{lem:sublinkExists} works when applying the weak Fr\'echet distance between a link and a subcurve. Therefore, \corref{cor:NRWFD} can follow from the aforementioned lemmas.
\end{proof}

\section{Strong NP-Hardness for Non-restricted GCS under the Hausdorff Distance} \label{sec:NRHD}

We show that all versions of undirected Hausdorff Distance are strongly NP-hard. This was shown for the vertex-restricted case by van Kreveld \etal~\cite{klw-oopsihfd-18} by a reduction from Hamiltonian cycle in segment intersection graphs. Their proof mostly extends straightforwardly to the curve-restricted and unrestricted case; however, because of the increased freedom in vertex placement we must take some care in the exact embedding of the segment graph: for instance, segments that intersect at arbitrarily small angles could potentially cause the reduction to produce coordinates with unbounded bit complexity. For this reason, we here reduce from a more restricted class of graphs: {\em orthogonal} segment intersections graphs. For completeness, we present the full adapted proof.

Czyzowicz~\etal~\cite {Czyzowicz:1998:SPR:292203.292211} show that Hamiltonian cycle remains NP-complete in 2-connected cubic bipartite planar graphs, and
Akiyama~\etal~\cite {Akiyama1980} prove that every bipartite planar graph has a representation as an intersection graph of orthogonal line segments.
Hence, Hamiltonian cycle in orthogonal segment intersection graphs is NP-complete.

Let $S$ be a set of $n$ horizontal or vertical line segments in the plane. We may assume the segment endpoints have integer coordinates that are linear in $n$ and that all intersections are proper intersections.
We further assume that the intersection graph of $S$ is connected (if not, it will not have a Hamiltonian cycle).
Set $\delta = \frac18$, and let $D \subset \R^2$ be the Minkowski sum of $S$ and a closed ball of radius $\delta$; that is, $D$ is the set of all points at distance at most $\delta$ from $S$. 

Let $P$ be an initial polyline that stays on the union of $S$, and covers all segments of $S$ many times. For instance, we may begin with a spanning tree of the intersection graph of $S$, traverse the tree using a depth-first search, and whenever we encounter a new segment we visit both endpoints before leaving through the intersection point with the next segment. Such a path has linear complexity. Then, we make a linear number of copies of this path and concatenate them, creating a path $P$ of $O(n^2)$ complexity.
Now, $P$ has every ordered permutation of the segments as a subsequence of its vertices, allowing us to freely walk over the arrangement of $S$ when creating an output polyline.

An output polyline $P'$ with Hausdorff distance at most $\delta$ to $P$ must in particular visit the $\delta$-disks around all endpoints of $S$, while staying inside $D$. Since no two such disks are visible to each other within $D$ unless they are endpoints of the same segment (see Figure~\ref {fig:somefigure}), we will need at least one more vertex in $P'$ for each segment, in order to switch to the next. This any output trajectory $P'$ will need to have at least $3n + 1$ vertices (the last vertex is the same as the starting vertex).

Clearly, if the intersection graph of $S$ has a Hamiltonian cycle, following this cycle will yield a solution using $3n + 1$ vertices by simply placing a vertex at each segment intersection and two intermediate vertices at each segment endpoint (see Figure~\ref {fig:somefigure}).
On the other hand, any solution that uses only $3n + 1$ vertices myst be of this form: if we visit any segment more than once (and place a vertex there) we must place at least four vertices on (close to) such a segment. 
Theorem~\ref {thm:NRHD} now follows.

\begin{figure}[tb]
	\centering
	\includegraphics[scale=1]{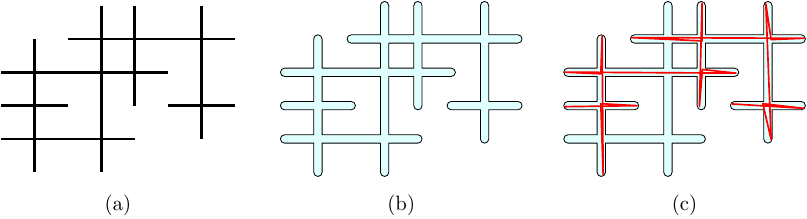}
	\caption{The construction: (a) the segments $S$, (b) the region $D$, (c) a partial simplification.}
	\label{fig:somefigure}
\end{figure}

\begin{theorem} \label{thm:NRHD}
	The non-restricted GCS problem under undirected Hausdorff distance is strongly NP-hard.
\end{theorem}
Since a solution to the reduction never benefits from placing vertices not on $P$, we also immediately obtain an improvement
for the special case $\Hausdorff_{\C}(P,\delta)$.

\begin{corollary} \label{cor:CRHD}
	The curve-restricted GCS problem under undirected Hausdorff distance is strongly NP-hard.
\end{corollary}

	\section{Conclusion}
	In this paper, we systematically studied the global curve simplification problem under different (global) distance measures and constraints in which vertices can be placed. We improved some of the existing results in the vertex-restricted case and we obtained the first NP-hardness results for the curve-restricted version. In the future, providing approximation algorithms, particularly for the curve-restricted case, can be of interest.

\bibliographystyle{abbrv}	
\bibliography{curvesimplification}

\end{document}